\documentclass[peerreview, draftcls, onecolumn, 12pt]{IEEEtran}
\ifx\pdfoutput\undefined
\bibliography{mybib}
\usepackage{graphicx}
\else
\usepackage[pdftex]{graphicx}
\fi
\usepackage{url}
\usepackage{array}
\usepackage{stfloats}
\usepackage{amssymb}
\usepackage{amsmath}
\usepackage{amsthm}
\usepackage{cite}
\usepackage{enumerate}
\usepackage{epstopdf}
\usepackage{color}
\DeclareGraphicsExtensions{.pdf,.png,.jpg}
\begin{document}
 \title{Constructive  Multiuser Interference in Symbol Level Precoding
for the MISO Downlink Channel}
 \author{
 \IEEEauthorblockN{ Maha~Alodeh,~\IEEEmembership{Student~Member, IEEE}, Symeon Chatzinotas, \IEEEmembership{Senior~Member,~IEEE,}~\newline
 Bj\"{o}rn Ottersten, \IEEEmembership{Fellow Member,~IEEE}\thanks{Maha Alodeh, Symeon Chantzinotas and  Bj\"{o}rn Ottersten are
with Interdisciplinary Centre for Security Reliability and Trust (SnT) at the University
of Luxembourg, Luxembourg. E-mails:\{ maha.alodeh@uni.lu, symeon.chatzinotas
@uni.lu, and bjorn.ottersten@uni.lu\}. \newline
This work is supported by Fond National de la Recherche Luxembourg (FNR)
projects,
project Smart Resource Allocation for Satellite Cognitive Radio (SRAT-SCR)  ID:4919957 and Spectrum Management and Interference Mitigation in Cognitive Radio Satellite Networks SeMiGod. \newline
Part of this work is accepted for publication in the proceedings of IEEE International Symposium on Information
Theory (ISIT), Honolulu-Hawaii, June 2014.}}\\
  
 } 
 
 \maketitle
 \IEEEpeerreviewmaketitle
 \begin{abstract}
\boldmath This paper investigates the problem of interference among the simultaneous
multiuser
transmissions in the downlink of multiple
antennas systems. Using symbol level precoding, a new approach towards the multiuser interference is discussed along
this paper. The concept of exploiting the interference between
the spatial multiuser transmissions by jointly utilizing the data information
(DI)
and channel state information (CSI), in order to design symbol-level precoders,
is proposed. In this direction, the interference among the data streams
is transformed under certain conditions to useful signal that can improve the signal to interference
noise ratio (SINR) of the downlink transmissions. We propose a maximum ratio
transmission (MRT) based algorithm that jointly exploits DI and CSI to glean the benefits from constructive multiuser interference. Subsequently, a relation between the constructive interference downlink transmission and physical
layer multicasting is established. In this context, novel constructive interference precoding techniques that tackle the transmit power minimization (min power) with
individual SINR constraints at each user's receivers is proposed. Furthermore, fairness through maximizing the weighted minimum SINR (max min SINR) of the users is addressed by finding the link between the min power and max min SINR problems. Moreover, heuristic precoding techniques are proposed
to tackle the weighted sum rate problem. 
 Finally, extensive numerical results
show that the proposed schemes outperform other state of the art techniques.\\

\end{abstract}

\vspace{-0.1cm}
\section{Introduction}
\IEEEPARstart{I}{nterference}
is one of the crucial and limiting factors in wireless networks. The idea of utilizing the time and frequency resources has been proposed in the literature to allow different users to share the resouces without inducing harmful interference. The
concept of exploiting the users' spatial separation has been a fertile research domain for more than one
decade. This can be implemented by adding multiple antennas at one or both
communication sides. Multiantenna transceivers empower the communication systems with more degrees
of freedom that can boost the performance if the multiuser interference is mitigated
properly. Exploiting the space dimension, to serve different users simultaneously
in the same time slot and the same frequency band through spatial division multiplexing (SDMA), has been investigated in \cite{roy}-\cite{caire}. 

The applications of SDMA, in which a single multiple antennas transmitter
wants to communicate with multiple receivers,
vary according to the requested
service. The first service type is known as a broadcast in which a transmitter
  has
a common message to be sent to multiple receivers. In physical layer research,  this service has been studied under the term of physical layer multicasting
 (i.e. \textit{PHY multicasting})
 \cite{multicast}-\cite{multicast-jindal}. Since a single data stream
is sent to all receivers, there is no multiuser interference.
In the remainder of this paper, this case will be referred to as multicast
transmission.
The second service type is known as unicast, in which a transmitter
has an individual message for each receiver. Due to the nature of the wireless medium
and the use of multiple antennas, multiple simultaneous unicast transmissions are possible in the downlink of a base station (BS). In these cases, multiple streams are
simultaneously sent, which motivates precoding techniques that mitigate the
multiuser interference. In information theory terms, this
service type has been studied using the broadcast channel \cite{caire}. In
the remainder of this paper, this case will be referred to as \textit{downlink} transmission.\smallskip

In the literature, the precoding techniques for downlink tranmission can be further classified as:
\begin{enumerate}
\item \textit{Group-level precoding} in which multiple codewords are transmitted simultaneously
but each codeword is addressed to a group of users. This case is also known as multigroup
multicast precoding \cite{g-multicast}-\cite{silva} and the precoder design
is dependant on the channels in each user group.
\item \textit{User-level precoding} in which multiple codewords are transmitted simultaneously
but each codeword is addressed to a single user. This case is also known as multiantenna
broadcast channel precoding \cite{mats}-\cite{ghaffar} and the precoder design
is dependant on the channels of the individual users. This is a special
case of group level precoding where each group consists of a single user.
\item \textit{Symbol-level precoding} in which multiple symbols are transmitted simultaneously
and each symbol is addressed to a single user 
\cite{Christos-1}-\cite{maha}. This is also known as a constructive interference
precoding and the precoder design is dependent on  both the channels (CSI) and
the symbols of the users (DI).
\end{enumerate} 
In the last category, the main idea is to constructively correlate
the interference among the spatial streams rather than fully decorrelate
them as in the conventional schemes \cite{haardt}. In \cite{Christos-1}, the interference in the scenario of BPSK and QPSK is classified into types: constructive and desctructive. Based on this classification, a selective channel inversion scheme is proposed to eliminate the destructive interference while it keeps the constructive one to be received at the users' terminal. A more advanced scheme is proposed in \cite{Christos}, which  rotates the destructive interference to be received as useful signal with the constructive one. These schemes outperform the conventional precodings \cite{haardt} and show considerable gains. However, the anticipated
gains come at the expense of additional complexity at the system design level. Assuming
that
the channel coherence time is $\tau_{c}$, and the symbol period is $\tau_s$, with $\tau_c\gg\tau_s$ for slow fading
channels, the user precoder has to be recalculated with a frequency of $\frac{1}{\tau_c}$
in comparison with the symbol based precoder $\frac{1}{\min(\tau_c,\tau_s)}=\frac{1}{\tau_s}$. Therefore, faster precoder calculation and switching is needed in
the symbol-level precoding which can be translated to more expensive hardware.  
The contributions of this paper can be summarized in the following points:

\begin{itemize}
\item  A generalized characterization of the constructive interference for any M-PSK is described. Based on this characterization, we propose a new constructive interference precoding scheme, called constructive interference maximum ratio transmissions (CIMRT). This technique exploits the weakness points of constructive interference zero forcing precoding (CIZF) in \cite{Christos}. 

\item We find the relation between the constructive interference precoding problem and PHY layer multicasting and verify it for any M-PSK modulation scenario. 

\item We propose different symbol based precoding schemes that aim at optimizing different performance metrics such as minimizing the transmit power while acheiving certain SNR targets, maximizing the minimum SNR among the user while keeping the power constraint in the system satisfied and finally maximizing the sum rate of all users without exceeding the permissible amount of power in the system.   


\end{itemize}
\vspace{-0.cm} 
 The rest of the paper is organized as follows: the channel and the system
 model is explained in section (\ref{system}), while section (\ref{traditional}) discusses
how the conventional downlink precoding techniques tackle interference. Symbol
level precoding is described in (\ref{constructive}). Moreover, techniques
that exploit the multiuser interference in symbol-based precoding are described  (\ref{masouros}). The relation to PHY-layer multicasting and the solution to the power min problem are investigated
in (\ref{powmin}). The problem of maximizing the minimum SINR is tackled in section (\ref{maxmins}). Heuristic
sum rate maximization techniques are discussed (\ref{wsr}). Finally, the performance
of the proposed
algorithms is evaluated in section (\ref{sim}).\smallskip 
  
\textbf{Notation}:  We use boldface upper and lower case letters for
 matrices and column vectors, respectively. $(\cdot)^H$, $(\cdot)^*$
 stand for Hermitian transpose and conjugate of $(\cdot)$. $\mathbb{E}(\cdot)$ and $\|\cdot\|$ denote the statistical expectation and the Euclidean norm,  $\mathbf{A}\succeq \mathbf{0}$ is used to indicate the positive
semidefinite matrix. $\angle(\cdot)$, $|\cdot|$ are the angle and magnitude  of $(\cdot)$ respectively. $\mathcal{R}(\cdot)$, $\mathcal{I}(\cdot)$
 are the real and the imaginary part of $(\cdot)$. Finally, the vector of
 all zeros with length of $K$ is defined as $\mathbf{0}^{K\times 1}$.
\vspace{-0.15cm} 
\section{System and Signal Models}
\label{system}
We consider a single-cell multiple-antenna downlink scenario,
where a single BS is equipped with $M$
transmit antennas that serves $K$ user terminals,
each one of them equipped with a single receiving antenna. The adopted
modulation technique is M-PSK.
We assume a quasi static block fading channel $\mathbf{h}_j\in\mathbb{C}^{1\times
M}$ between
the BS antennas and the $j^{th}$ user, where the received signal at
j$^{th}$ user is
written as
\vspace{-0.3cm}
\begin{eqnarray}
y_j[n]&=&\mathbf{h}_j\mathbf{x}[n]+z_j[n].
\end{eqnarray} $\mathbf{x}[n]\in\mathbb{C}^{M\times 1}$ is the transmitted signal vector from the multiple antennas
transmitter and  $z_j$ denotes the noise at $j^{th}$ receiver, which is assumed i.d.d  complex Gaussian distributed variable $\mathcal{CN}(0,1)$. A compact formulation
of the received signal at all users' receivers can be written as
\vspace{-0.1cm}
\begin{eqnarray}
\mathbf{y}[n]&=&\mathbf{H}\mathbf{x}[n]+\mathbf{z}[n].
\end{eqnarray}
Let $\mathbf{x}[n]$ be written as $\mathbf{x}[n]=\sum^K_{j=1}\mathbf{w}_j[n]d_j[n]$,
where $\mathbf{w}_j$ is the $\mathbb{C}^{M\times
1}$ unit power precoding vector for the user $j$. The received signal at $j^{th}$
user ${y}_j$ in $n^{th}$ symbol period is given by
\begin{eqnarray}
\label{rx_o}
{y}_j[n]=\sqrt{p_j[n]}\mathbf{h}_j\mathbf{w}_j[n] d_j[n]+\displaystyle\sum_{k\neq j}\sqrt{p_k[n]}\mathbf{h}_j\mathbf{w}_k[n]
d_k[n]+z_j[n]
\end{eqnarray}
where $p_j$ is the allocated power to the $j^{th}$ user. A more detailed compact system formulation
is obtained by stacking the received signals and the noise
components for the set of K selected users as
\begin{eqnarray}
\mathbf{y}[n]=\mathbf{H}\mathbf{W}[n]\mathbf{P}^{\frac{1}{2}}[n]\mathbf{d}[n]+\mathbf{z}[n]
\end{eqnarray}
with $\mathbf{H} = [\mathbf{h}_1,..., \mathbf{h}_K]^T \in\mathbb{C}^{K\times M} $, $\mathbf{W}=[\mathbf{w}_1, ...,\mathbf{w}_K]\in\mathbb{C}^{nt\times M}$ as the
compact channel and precoding matrices. Notice that the transmitted signal $\mathbf{d}\in\mathbb{C}^{K\times 1}$
includes the uncorrelated data symbols $d_k$ for all users with $\mathbb{E}[{|d_k|^2}] = 1$, $\mathbf{P}^{\frac{1}{2}}[n]$
is the power allocation matrix $\mathbf{P}^{\frac{1}{2}}[n]=diag(\sqrt{p_1[n]},\hdots,\sqrt{p_K[n]})$.
It should be noted that CSI and DI are available at the transmitter side. 
\vspace{-0.2cm}
\section{Conventional Multiuser Precoding Techniques}
\label{traditional}
The main goal of transmit beamforming is to increase the signal power at
the intended user and mitigate the interference to non-intended users. This
can be obtained by precoding the transmitted symbols in a way that optimizes
the spatial directions of the simultaneous transmissions by means
of  beamforming. This can be mathematically translated to a design problem that targets
 beamforming vectors to have maximal inner products with 
 the intended channels and minimal inner products with the non-intended user channels. There are several proposed beamforming techniques in the literature.
One of the simplest approaches is to encode the transmitted
signal by pre-multiplying it with the pseudo inverse of the multiuser matrix
channel.
Several approaches have been
proposed including minimizing the sum power while satisfying
 a set of SINR constraints\cite{mats} and  maximizing the jointly achievable SINR
margin under a  power constraint\cite{boche}. 
In any scenario,  
the generic received signal can be formulated as 
\begin{eqnarray}\nonumber
\label{interference}
\mathbf{y}[n]&=&\mathbf{H}\mathbf{x}[n]+\mathbf{z}[n]=\mathbf{H}\mathbf{W}[n]\mathbf{P}^{\frac{1}{2}}[n]\mathbf{d}[n]+\mathbf{z}[n]\\
&=&\begin{bmatrix}\underset{\text{desired}}{\underbrace{a_{11}}}&\underset{\text{interference}}{\underbrace{a_{12}}}&\hdots&\underset{\text{interference}}{\underbrace{a_{1K}}}\\
a_{21}&\underset{\text{desired}}{\underbrace{a_{22}}}&\dots&a_{2K}\\
\vdots&\vdots&\vdots&\vdots\\
\underset{\text{interference}}{\underbrace{a_{K1}}}&\underset{\text{interference}}{\underbrace{a_{K2}}}&\dots&\underset{\text{desired}}{\underbrace{a_{KK}}}
\end{bmatrix}\begin{bmatrix}d_1\\
\vdots\\
d_K\end{bmatrix}+\mathbf{z}.
\end{eqnarray}
The corresponding SINR can be expressed as\\
\begin{eqnarray}
\gamma_{j}=\frac{p_k\|\mathbf{h}_j\mathbf{w}_k\|^2}{\sum^K_{i=1,i\neq k}p_i\|\mathbf{h}_j\mathbf{w}_i\|^2+\sigma^2}=\frac{|a_{jj}|^2}{\sum^K_{i=1,i\neq
k}|a_{ji}|^2+\sigma^2}.
\end{eqnarray}

This paper tries to go beyond this conventional look at the interference by employing symbol-level precoding.
This approach can under certain conditions convert the inner product 
with the non-intended channels into useful power by maximizing them but
with the specific directions to which constructively add-up at each user receivers.
Taking into account  the I/Q plane of the symbol detection, the constructive interference is achieved by using the interfering signal vector to move the received point deeper into the correct detection region. Considering
that each user receives a constructive interference from other users' streams,
the received signal can be written as
\vspace{-0.1cm}
\begin{eqnarray}
y_j[n]=\sum^K_{i=1}\underset{a_{ji}[n]d_{k}[n]}{\underbrace{\mathbf{h}_j\mathbf{w}_i[n]d_i[n]}}+z_k[n].
\end{eqnarray}
 This yields the SINR expression for M-PSK symbols as 
\begin{eqnarray}
\gamma_{k}[n]=\frac{\|\sum^K_{i=1}\mathbf{h}_j\mathbf{w}_i[n]\|^2}{\sigma^2}=\frac{|\sum^K_{i=1}a_{ji}|^2}{\sigma^2}.
\end{eqnarray}
Different precoding techniques that redesign the terms $a_{ji},j\neq i$ to constructively
correlate them with $a_{jj}$ are proposed in the next sections (\ref{masouros})-(\ref{wsr}).
\vspace{-0.2cm}
\subsection{Power constraints for user based and symbol based precodings}
In the conventional user based precoding, the transmitter needs to precode every $\tau_{c}$
which means that the power constraint has to be satisfied along the coherence time
$\mathbb{E}_{\tau_c}\{\|\mathbf{x}\|^2\}\leq
P$. Taking the expectation of $\mathbb{E}_{\tau_c}\{\|\mathbf{x}\|^2\}=\mathbb{E}_{\tau_c}\{tr(\mathbf{W}\mathbf{d}\mathbf{d}^H\mathbf{W}^H)\}$,
and since $\mathbf{W}$ is fixed along $\tau_c$, the previous expression can
be reformulated as $tr(\mathbf{W}\mathbb{E}_{\tau_c}\{\mathbf{d}\mathbf{d}^H\}\mathbf{W}^H)=tr(\mathbf{W}\mathbf{W}^H)=\sum^K_{j=1}\|\mathbf{w}_j\|^2$,
where $\mathbb{E}_{\tau_c}\{\mathbf{d}\mathbf{d}^H\}=\mathbf{I}$ due to uncorrelated
symbols over $\tau_c$.

However, in symbol level precoding the power constraint should be guaranteed
for each symbol vector transmission namely for each $\tau_s$. In this case
the power constraint equals to $\|\mathbf{x}\|^2=\mathbf{W}\mathbf{d}\mathbf{d}^H\mathbf{W}^H=\|\sum^K_{j=1}\mathbf{w}_jd_j\|^2$.
In the next sections, we characterize the constructive interference and show
how to exploit it in the multiuser downlink transmissions\footnote{From now on, we assume that the transmssion changes at each symbol and we drop the time index for the ease of notation }. 

 \vspace{-0.3cm}
       
 \section{Constructive Interference}
 \label{constructive}
 \vspace{-0.1cm}
The interference is a random deviation which can move the desired constellation point in any direction. To address this problem, the power of the interference has been used in the past to regulate its effect on the desired signal point. 
The interference among the multiuser spatial streams
leads to deviation of the received symbols outside of their detection region. However,
in symbol level precoding (e.g. M-PSK) this interference pushes the received symbols further into the correct detection region and, as a consequence it enhances the system performance. Therefore, the interference can
be classified into constructive or destructive based on whether it facilitates or deteriorates the correct detection of the received symbol. For BPSK and QPSK scenarios, a detailed classification of interference is discussed thoroughly in \cite{Christos-1}. In this section, we describe the required conditions to have constructive interference for any M-PSK modulation.
\vspace{-0.1cm}
\subsection{Constructive Interference Definition}
Assuming both DI and CSI are available at the transmitter, the unit-power
 created interference from the $k^{th}$ data stream on $j^{th}$ user can be formulated as:
\vspace{-0.1cm}
\begin{equation}
\psi_{jk}=\frac{\mathbf{h}_{j}\mathbf{w}_k}{\|\mathbf{h}_{j}\|\|\mathbf{w}_k\|}.
\end{equation}
Since the adopted modulations are M-PSK ones, a definition for
constructive interference can be stated as\smallskip

\begin{newtheorem}{lemma}{\textbf{Lemma}}
\begin{lemma}
\label{lemma}
For any M-PSK modulated symbol $d_k$, it is said to receive constructive
interference from another simultaneously transmitted symbol $d_j$ which is
associated with $\mathbf{w}_j$ if and only if the following inequalities hold   
\begin{equation}\nonumber
\label{one}
\angle{d_j}-\frac{\pi}{M}\leq \arctan\Bigg(\frac{\mathcal{I}\{\psi_{jk}d_{k}\}}{\mathcal{R}\{\psi_{jk}d_{k}\}}\Bigg)\leq \angle{d_j}+\frac{\pi}{M},
\end{equation}
\begin{equation}\nonumber
\label{two}
\mathcal{R}\{{d_k}\}.\mathcal{R}\{\psi_{jk}
d_{j}\}>0, \mathcal{I}\{{d_k}\}.\mathcal{I}\{\psi_{jk}d_{j}\}>0.\\
\end{equation}
\end{lemma}\smallskip
 \vspace{0.1cm}
\begin{proof}
For any M-PSK modulated symbol, the region of correct detection lies in $\theta_j\in[\angle d_j-\frac{\pi}{M},\angle d_j+\frac{\pi}{M}]$, where $\theta_j$ is the angle of the detected symbols.
In order for the interference to be constructive, the received interfering signal
should lie in the region of the target symbol. For the first condition, the
${\arctan}(\cdot)$ function checks whether the received interfering signal
originating from the $d_k$$^{th}$ transmit symbol is located in
the detection region of the target symbol. However, the trigonometric
functions are not one-to-one functions. This means that it manages to check
the two quadrants which the interfering symbol may lie in. To find which
one of these quadrants is the correct one, an additional constraint is added to check the sign compatibility
of the target and received interfering signals. 
\end{proof}
\end{newtheorem}

\begin{newtheorem}{cor}{\textbf{Corollary}}
\begin{cor}
The constructive interference is mutual.
If the symbol $d_j$ constructively interferes with $d_k$, then
the interference from transmitting  the symbol $d_k$ 
is constructive to $d_j$.\\ 
\end{cor}
%
\end{newtheorem}

For constructively interfering symbols, the value of the received signal can be bounded as
\vspace{-0.1cm}
\begin{eqnarray}
\vspace{-0.1cm}
\label{Rxs}
\sqrt{p}_j\|\mathbf{h}_j\|\overset{(a)}{\leq} |y_j|\overset{(b)}{\leq} \|\mathbf{h}_{j}\|\big(\sqrt{p}_j+\displaystyle\sum^{K}_{\forall
k,k\neq j}\sqrt{p}_k|\psi_{jk}|\big).
\vspace{-0.5cm}
\end{eqnarray}
The inequality (a) holds when all simultaneous users are orthogonal (i.e. $\psi_{jk}=0$), while (b) holds when all created interference
is aligned with the transmitted symbol as $\angle d_k=\angle\psi_{jk} d_j$
 and $\psi_{jk}=0$, $\angle d_k=\angle\psi_{jk} d_j$.  
Eq. (\ref{Rxs}) indicates that in the case of constructive
interference, having fully correlated signals is beneficial as they contribute
to received signal power. 
In conventional precoding techniques, the previous inequality
can be reformulated as
\begin{eqnarray}
\label{rxs}
\vspace{-0.6cm}
0\overset{(a)}{\leq} |y_j| \overset{(b)}{\leq} \sqrt{p_j}\|\mathbf{h}_j\|.
\vspace{-0.5cm}
\end{eqnarray}
The worst case scenario can occur when all users are co-linear $\psi_{jk}\rightarrow
1$. The channel cannot
be inverted and thus the interference cannot be mitigated. The optimal scenario takes place
when all users have physically orthogonal channels which entails no multiuser interference. Therefore, utilizing the CSI
and DI leads to higher performance in comparison with conventional techniques.

\vspace{-0.2cm}
\section{Constructive Interference Precoding for MISO Downlink Channels }
\label{masouros}
In the remainder of this paper, it is assumed that the transmitter is capable
of designing precoding on symbol level  utilizing both
CSI and DI\footnote{From this section, we combine the the precoding design with power allocation}.
\vspace{-0.1cm}    
\subsection{Correlation Rotation Zero Forcing Precoding (CIZF)}

The precoder aims at minimizing the mean square error while it takes into
the account the rotated constructive interference \cite{Christos}. The optimization problem
can be formulated as
\vspace{-0.25cm}    
\begin{eqnarray}\nonumber
\mathcal{J}=\min_{\mathbf{W}}\quad\mathbb{E}\{\Vert\mathbf{R}_{\phi}\mathbf{d}-(\mathbf{H}\mathbf{W}\mathbf{d}+\mathbf{z})\Vert^2\},
\end{eqnarray}
where $\mathbf{P}^{\frac{1}{2}}=\mathbf{I}$ in this scenario. The solution can be easily expressed as 
\begin{eqnarray}
\label{CIZF}
\vspace{-0.2cm}
\mathbf{W}_{CIZF}=\gamma\mathbf{H}^H(\mathbf{H}\mathbf{H}^H)^{-1}\mathbf{R}_{\phi},
\end{eqnarray}
where $\gamma=\sqrt{\frac{P}{tr\big(\mathbf{R}^H_{\phi}(\mathbf{H}\mathbf{H}^H)^{-1}\mathbf{R}_{\phi}\big)}}$
ensures the power normalization. The cross correlation
factor between the $j^{th}$ user's channel and transmitted $k^{th}$ data
stream can be expressed as
\vspace{-0.1cm}
\begin{eqnarray}
\rho_{jk}=\frac{\mathbf{h}_j\mathbf{h}^H_k}{\|\mathbf{h}_k\|\|\mathbf{h}_j\|}.
\end{eqnarray}
The relative phase $\phi_{jk}$ that grants the constructive simultaneous
transmissions can be expressed as 
\vspace{-0.1cm}
\begin{eqnarray}
\phi_{jk}=\angle d_j-\angle(\rho_{jk}.d_k).
\end{eqnarray}
The corresponding rotation matrix can be implemented as:
\vspace{-0.2cm}
\begin{eqnarray}
\mathbf{R}_{\phi} (j,k)=\rho_{jk}\exp(\phi_{jk}i),
\end{eqnarray}
and the received signal at $j^{th}$ user can be expressed as
\vspace{-0.25cm}
\begin{eqnarray}
\label{mas}
y_j\overset{a}{=}{\gamma}{{\|\mathbf{h}_j\|(\sum^K_{k=1}\rho_{jk}d_k)}}\overset{b}{=}{\gamma}\|\mathbf{h}_j\|(\sum^K_{k=1}\varepsilon_{jk})d,
\end{eqnarray}
where $\varepsilon_{jk}$ has the same magnitude as $\rho_{jk}$ but with different
phase, and $d:d \in \mathbb{C}^{1\times 1},|d|=1, \angle d=\theta,\theta
\in [0,2\pi] $.
By taking a look at (\ref{mas}-b), it has a multicast formulation since
it seems for each user that BS sends the same symbol for all users by applying a user-dependent rotation.. 

\begin{newtheorem}{rem}{\textbf{Remark}}
\begin{rem}
It can be noted that this solution includes a zero forcing step and a correlation
step $\mathbf{R}_\phi$. The correlation step aims at making
the transmit signals constructively received at each user.
Unfortunately, this design fails when we deal with co-linear users $\rho_{jk}\rightarrow
1$. However, intuitively
having co-linear users should create more constructive interference and higher
gain should be anticipated. It can be easily concluded that the source of
this
contradiction is the zero forcing step. In an effort to overcome the problem, we propose a new precoding technique in the next section.   
\end{rem}
\end{newtheorem}
\vspace{-0.1cm}
\subsection{Proposed Constructive Interference Maximum Ratio Transmission
(CIMRT)}

The maximum ratio transmission (MRT) is not suitable for multiuser downlink
transmissions in MISO system due to the intolerable amount of the created interference. On the other hand, this feature makes it a good candidate for
constructive interference. The \textit{naive}
maximum ratio transmission (nMRT) can be formulated as 
\begin{eqnarray}
\mathbf{W}_{\text{nMRT}}=\begin{bmatrix}\frac{{\mathbf{h}_1}^H}{\|\mathbf{h}_1\|}, \frac{{\mathbf{h}_2}^H}{\|\mathbf{h}_2\|},\hdots,\frac{{\mathbf{h}_K}^H}{\|\mathbf{h}_K\|}\end{bmatrix}.
\end{eqnarray}
A new look at the received signal can be viewed by exploiting the singular
value decomposition of $\mathbf{H}=\mathbf{S}\mathbf{V}\mathbf{D}$, and $\mathbf{W}_{nMRT}=\mathbf{D}^H\mathbf{V}^{'}\mathbf{S}^H$ as follows
\vspace{-0.1cm}
\begin{eqnarray}\nonumber
\label{svd}
\vspace{-0.2cm}
\mathbf{y}&=&\mathbf{H}\mathbf{W}\mathbf{d}+\mathbf{z}={{\mathbf{S}\mathbf{V}\mathbf{D}\mathbf{D}^H\mathbf{V}^{'}}}{{\mathbf{S}^H}}\mathbf{d}+\mathbf{z}\\
\quad&=&\mathbf{G}\mathbf{B}\mathbf{d}+\mathbf{z},
\end{eqnarray} 
where
\begin{eqnarray}\nonumber
\mathbf{G}&=&\mathbf{S}\mathbf{V}\mathbf{V}^{'},\quad \mathbf{B}=\mathbf{S}^H.
\end{eqnarray}
$\mathbf{S}\in \mathbb{C}^{K\times K}$ is a unitary matrix that contains
the
left-singular vectors of
$\mathbf{H}$, the matrix $\mathbf{V}$ is an $K\times M$ diagonal matrix with nonnegative real numbers on the diagonal, and $\mathbf{D}\in\mathbb{C}^{nt\times nt}$
contains right-singular vectors of $\mathbf{H}$. $\mathbf{V}^{'}$ is the power scaled of $\mathbf{V}$
to normalize each column in $\mathbf{W}_{\text{nMRT}}$ to unit.
The received signal can be as
\vspace{-0.1cm}
\begin{eqnarray}
y_j=\|\mathbf{h}_j\|(\sum^K_{k=1}\sqrt{p_k}\rho_{jk}d_k)+z_j.
\end{eqnarray}
\vspace{-0.05cm}
Utilizing the reformulation of $\mathbf{y}$ in
(\ref{svd}), the received signal can be written as 
\vspace{-0.05cm}
\begin{eqnarray}
\label{rot}
y_j=\|\mathbf{g}_j\|\sum^K_{k=1}\sqrt{p_k}\xi_{jk}d_k=\|\mathbf{g}_j\|\sum^K_{k=1}\sqrt{p_k}\xi_{jk}\exp(\theta_k)d
\end{eqnarray}
where $\mathbf{g}_j$ is the $j^{th}$ row of the matrix $\mathbf{G}$, $\xi_{jk}=\frac{\mathbf{g}_j\mathbf{b}_k}{\|\mathbf{g}_j\|}$. Since
$\mathbf{B}$ is a unitary matrix, it can have uncoupled rotations which can
grant the constructivity of interference.
Let $\mathbf{R}_{kj}$ be the rotation matrix in the
$(\mathbf{b}_k,\mathbf{b}_j)$-plane, which performs an orthogonal rotation of
the $k^{th}$ and $j^{th}$ columns of a unitary matrix while keeping
the others fixed, thus preserving unitarity. Assume without
loss of generality that $k >j $. Givens rotation matrix in the
$(\mathbf{b}_k,\mathbf{b}_j)$-plane can be defined as

\scriptsize
\begin{eqnarray}\nonumber
\label{rotation}
\hspace{-0.5cm}\begin{array}{cccc}
\small\mathbf{R}_{kj}(\alpha,\delta)=\small\small\begin{bmatrix}\begin{array}{cccccc}1&0&\hdots&0&\hdots&0\\
\vdots&\ddots&\vdots&\vdots&\vdots&\vdots\\
0&\hdots&\cos\alpha&\hdots&-\sin\alpha e^{-\delta i}&\hdots\\
\vdots&\vdots&\vdots&\vdots&\vdots\\
\vdots&\vdots&\sin\alpha e^{-\delta i}&\hdots&\cos\alpha&\hdots\\
\vdots&\hdots\hdots&\vdots&\hdots&&1\end{array}\end{bmatrix}\hspace{-0.5cm}
\end{array}
\end{eqnarray}
\normalsize
where the non trivial entries appear at the intersections of
$k^{th}$ and $j^{th}$ rows and columns. Hence, any unitary matrix
$\mathbf{B}^{'}$ can be expressed using the following parameterization
\begin{eqnarray}
\mathbf{B}^{'}=\mathbf{B}\prod^K_{j=1}\prod^K_{k=j+1}\mathbf{R}_{kj}.
\end{eqnarray}
It can be seen from the structure of the matrix in (\ref{rotation}) that
rotation in the ($\mathbf{b}_k$,$\mathbf{b}_j$)-plane does not change the directions
of the remaining beamforming vectors. Therefore, it just modifies the value
of $\xi_{kk}$, and the precoder reads as
\vspace{-0.1cm}
\begin{eqnarray}
\label{CIMRT}
\mathbf{W}_{CIMRT}&=&\mathbf{D}^H\mathbf{V}^{'}\mathbf{B}^{'}.
\end{eqnarray}
To grant constructive interference, we need to rotate the ($\mathbf{b}_k$,$\mathbf{b}_j$)-plane by formulating the rotation as a set of non-linear equations as 
\vspace{-0.15cm}
\begin{eqnarray}\nonumber
\label{rotate}
\xi^{'}_{kk}d_k&=&\xi_{kk}\cos(\alpha) d_k-\xi_{kj}\sin(\alpha) e^{-i\delta}d_j\\
\xi^{'}_{jj}d_j&=&\xi_{jk}\sin(\alpha) e^{i\delta}d_k+\xi_{jj}\cos(\alpha) d_j.
\end{eqnarray}
\begin{rem}
The rotation of $(k,j)$ plane is independent and decoupled from any other
plane. This means that any implemented rotation on this plane only affects the $(k,j)$ pair.\\  
\end{rem}
Since the set of non-linear equations can have different roots, the function
needs to be evaluated at the obtained root in order to find the optimal ones.
The optimal solution can be found when
solving for $\xi^{'}_{kk}=\sqrt{\xi^2_{kk}+\xi^2_{kj}}$, $\xi^{'}_{jj}=\sqrt{\xi^2_{jj}+\xi^2_{jk}}$.
Sometimes it is not feasible to solve for $\xi^{*}_{kk}$ and $\xi^{*}_{jj}$,
and
their values need to be reduced correspondingly. The proposed algorithm can be illustrated in the following table 
\hspace{-0.4cm}\begin{center}
    \begin{tabular}{ p{8.3cm} }
    \hline
   \footnotesize \textbf{A1}: Constructive Interference Rotation for CIMRT Algorithm\\
    \hline
    \vspace{-0.3cm}
    \begin{enumerate}
    \item    \footnotesize Find $\mathbf{P}$ assuming all the users have constructive interference.
    \item Find singular value decomposition for $\mathbf{H}=\mathbf{S}\mathbf{V}\mathbf{D}$.
    \item Construct $\mathbf{B}$, $\mathbf{G}$.
    \item for $i\in\forall(k,j)$ combinations
    \begin{enumerate}
    \item Select ($\mathbf{b}_k$, $\mathbf{b}_j$)-plane.  
    \item Find the optimal rotation parameters $\alpha$, $\delta$  for ($\mathbf{b}_k$,$\mathbf{b}_j$)
    considering $\mathbf{P}$ by solving (\ref{rotate}).
    \item Update $\mathbf{B}=\mathbf{B}\mathbf{R}_{kj}(\alpha,\delta)$.
    \end{enumerate}
    end
    \item The final precoder $\mathbf{W}=\mathbf{D}\mathbf{V}^{'}\mathbf{B}\prod^K_{j=1}\prod^K_{k=j+1}\mathbf{R}_{kj}$
    
    \end{enumerate}\\
    \hline
\end{tabular}
\end{center}

 \section{Constructive Interference for Power Minimization}
 \label{powmin}
 \subsection{Constructive Interference Power Minimization Precoding (CIPM)}

  From the definition of constructive interference, we should design the constructive interference precoders by granting that the sum of the
precoders and data symbols in the expression  forces the received
signal to the detection region of the desired symbol for each user. 
 Therefore, the optimization that
 minimizes the transmit power and grants
 the constructive reception of the transmitted data symbols can be written
 as 
\vspace{-0.2cm}
\begin{eqnarray}
\label{powccm}
\hspace{-0.2cm}\mathbf{w}_k(d_j,\mathbf{H},\boldsymbol\zeta)
\hspace{-0.1cm}&=&\arg\underset{\mathbf{w}_1,\hdots,\mathbf{w}_K}{\min}\quad \|\sum^K_{k=1}\mathbf{w}_kd_k\|^2\\\nonumber
\hspace{-0.1cm}&s.t.&\begin{cases}\mathcal{C}1:\angle(\mathbf{h}_j\sum^K_{k=1}\mathbf{w}_k
d_k)=\angle(d_j),
\forall j\in K\\
\mathcal{C}2:\|\mathbf{h}_j\sum^K_{k=1}\mathbf{w}_kd_k\|^2\geq\sigma^2\zeta_j\quad
, \forall j\in K,
\end{cases}
\end{eqnarray}\\
where $\zeta_j$ is the SNR target for the $j^{th}$ user that should
 be granted by the transmitter, and ${\boldsymbol\zeta}=[\zeta_1,\hdots,\zeta_K]$ is the vector that contains all the SNR targets.  The set of constraints $\mathcal{C}_1$
 guarantees that each user receives its corresponding data symbol $d_j$. A reformulation for the previous problem
 (\ref{powccm}) using $\mathbf{\hat{w}}_k=\mathbf{w}_kd_k$ can be expressed as

\hspace{-0.7cm} \begin{eqnarray}
\label{powccc}
\mathbf{\hat{w}}_k(d_j,\mathbf{H},\boldsymbol\zeta)
\hspace{-0.3cm}&=&\hspace{-0.1cm}\arg\underset{\mathbf{\hat{w}}_1,\hdots,\mathbf{\hat{w}}_K}{\min}\quad \|\sum^K_{k=1}\mathbf{\hat{w}}_k\|^2\\\nonumber
&s.t.&\begin{cases}\mathcal{C}1:\angle(\mathbf{h}_j\sum^K_{k=1}\mathbf{\hat{w}}_k)=\angle(d_j),
\forall j\in K\\
\mathcal{C}2:\|\mathbf{h}_j\sum^K_{k=1}\mathbf{\hat{w}}_k\|^2\geq\sigma^2\zeta_j\quad
, \forall j\in K.
\end{cases}
\end{eqnarray}\\
The replaced variables $\mathbf{\hat{w}}_k$'s indicate that it is not necessary
to send the exact symbols $d_1, \hdots,d_K$; they can be included in precoding design as long as they are
received correctly at users' terminals. Then, we design the final output
vector
$\mathbf{x}$ instead of designing the whole $\mathbf{W}$ with the assumption that
$\mathbf{d}$ is fixed.This means that the proposed methods move away from the classic approach of linear beamforming, where the precoding matrix is multiplied with the symbol vector. Instead, we adopt an approach  where the transmit signal vector is designed directly based on an optimization problem.
\smallskip
\begin{lemma}
Assuming a conventional linear precoder $\mathbf{x}=\mathbf{W}\mathbf{d}$, the transmitted signal vector $\mathbf{x}$ which minimizes the transmit power can be calculated using a unit-rank precoding matrix $\mathbf{W}$.

\begin{proof}
This can be proved by using the auxiliary variable $\mathbf{x}=\sum^K_{k=1}\mathbf{\hat{w}}_k$
and substituting it in the optimization problem (\ref{powccc}). The optimization
can be rewritten as
\begin{eqnarray}\nonumber
\label{unitrank}
\hspace{-0.1cm}\mathbf{x}(d_j,\mathbf{H},\boldsymbol\zeta)
&=&\arg\underset{\mathbf{x}}{\min}\quad \|\mathbf{x}\|^2\\\nonumber
&s.t.&\begin{cases}\mathcal{C}1:\angle(\mathbf{h}_j\mathbf{x})=\angle(d_j),
\forall j\in K\\
\mathcal{C}2:\|\mathbf{h}_j\mathbf{x}\|^2\geq\sigma^2\zeta_j\quad
, \forall j\in K.
\end{cases}
\end{eqnarray}\\
$\mathbf{x}$ is a vector, which makes the solution $\mathbf{W}=[\mathbf{w}_1,\hdots,\mathbf{w}_K]$ a unit rank as $\mathbf{w}_k=\frac{\mathbf{x}}{K}$, and the virtual input
vector $\mathbf{d}_v=\mathbf{1}^{K\times 1}$. 
\end{proof} 
\end{lemma}
Based on Lemma 2, the differentiation between the conventional and constructive
interference precoding techniques is illustrated in Fig. (\ref{Tprecoding})-(\ref{Cprecoding}).
Fig. (\ref{Tprecoding}) shows how the conventional precoding depends only on the CSI information to optimize $\mathbf{W}$ that
carry the data symbols $\mathbf{d}$ and without any design dependency
between them. Therefore, the transmitted output vector can be formulated as $\mathbf{x}=\sum^K_{j=1}\mathbf{w}_jd_j$.
The final output vector $\mathbf{x}$ only depends on the DI and CSI and this dependence is a
linear one. On the other hand, in constructive interference precoding schemes,  the precoding directly depends on both the CSI and DI information to exploit
the interference through skipping the intermediate step (i.e. optimizing
$\mathbf{W}$) and optimizing directly the vector  $\mathbf{x}$.
In constructive interference schemes (\ref{unitrank}), the relation between
the data symbols in $\mathbf{d}$ and the final output vector $\mathbf{x}$ cannot be explicitly
described as in linear conventional precoding scheme. This can be explained by the
fact that the DI
is used to design the output vector but is not necessarily physically transmitted as in
conventional linear precoding. An implicit set of virtual data is used instead which is explained
later in this paper. 
 
\begin{figure*}[t]
\vspace{-1.5cm}
\hspace{0.2cm}
\begin{tabular}[t]{c}
\begin{minipage}{17 cm}
 \begin{center}
\hspace{-.5cm}\includegraphics[scale=0.5]{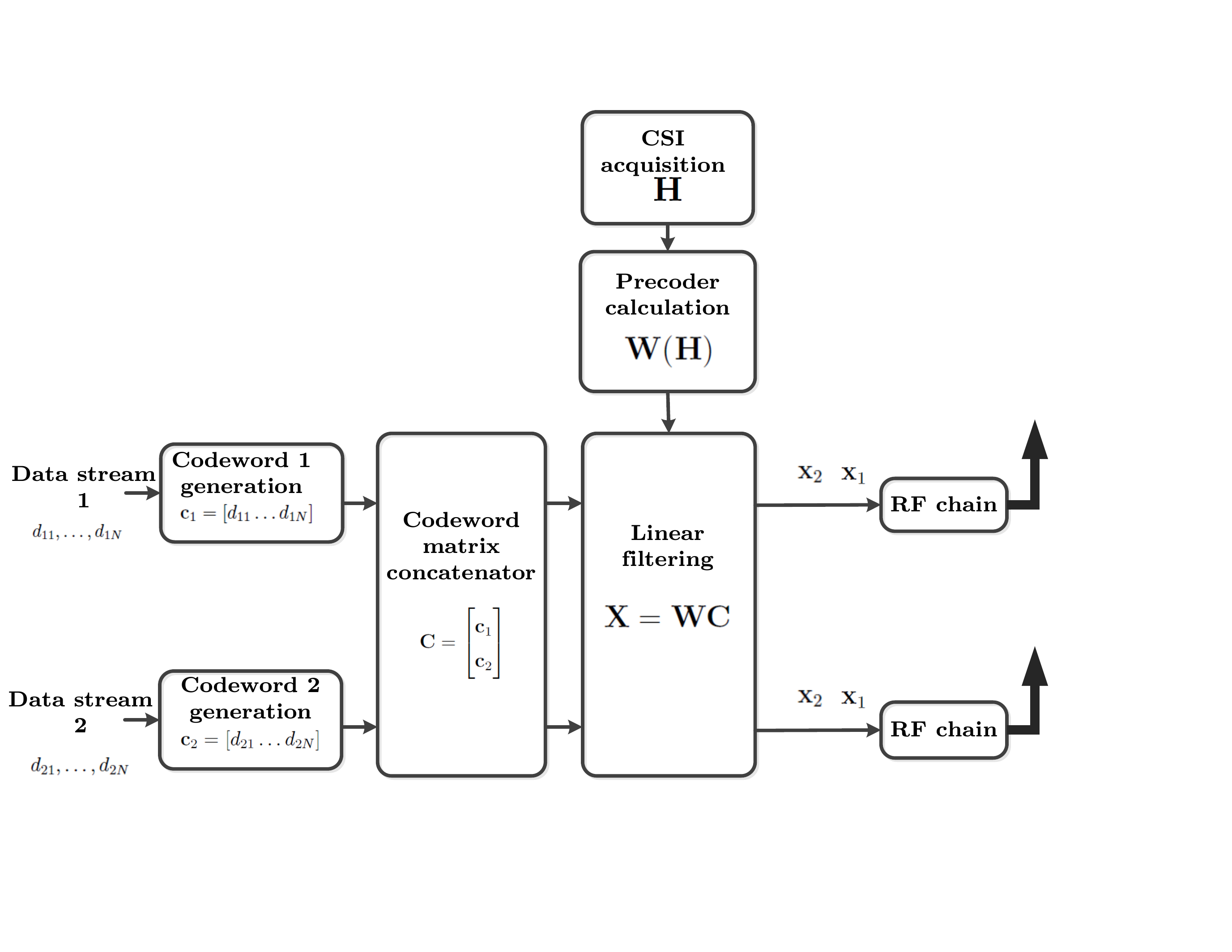}
\vspace{-1.3cm}\caption{\label{Tprecoding}Codeword-level precoding model in the conventional
MISO systems. The precoder is calculated and applied once for the whole codeword since it is independent of the actual symbols. }
\end{center}
\end{minipage}\\
\end{tabular}
\end{figure*}
\begin{figure} 
\vspace{-1.5cm}

\hspace{-2.5cm}\includegraphics[scale=0.5]{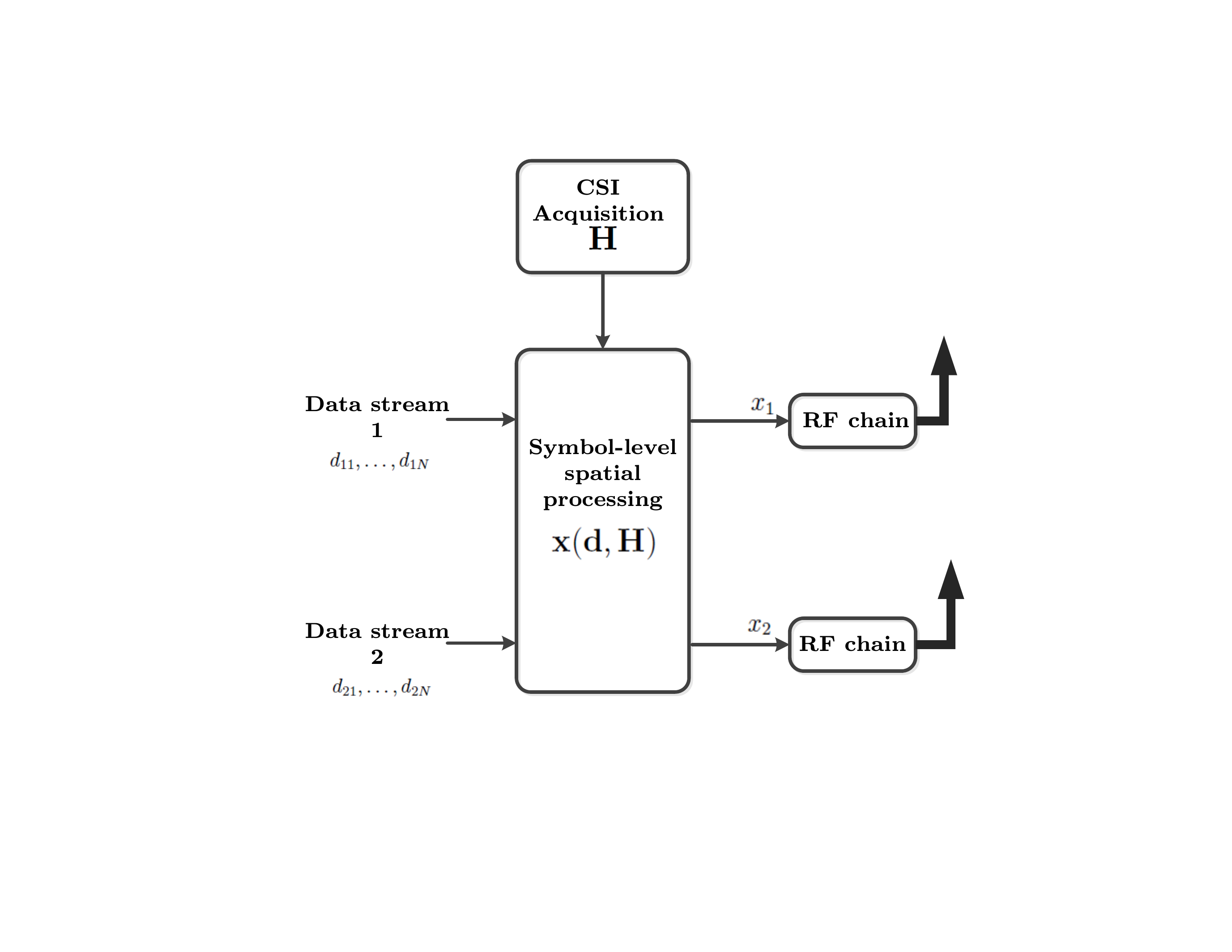}
\vspace{-3.0cm}\caption{\label{Cprecoding}Symbol-level precoding model in the constructive
interference MISO systems. The transmit vector $\mathbf{x}$ is calculated once per symbol. }
\end{figure}

 \subsection{The Relation Between Constructive Interference Precoding and Constrained
Constellation Multicast}
By taking a look at Lemma 2, the solution of the optimization problem resembles
the solution of multicast problem in which the transmitter sends a single message to
multiple users\cite{multicast}-\cite{jorswieck}. However in our problem,
we have an additional constraint $\mathcal{C}_1$ in (\ref{powccm})-(\ref{unitrank}) which guarantees that each user detects correctly its symbol based on the received signal.\\

\begin{newtheorem}{thm1}{\textbf{Theorem}} 
\begin{thm1}
\label{thmo}
The optimal precoder for CIPM 
\vspace{-0.1cm}
\begin{eqnarray}\nonumber
\label{powcd}
\hspace{-0.7cm}
\mathbf{x}_{CIPM}(\mathbf{d},\mathbf{H},\boldsymbol\zeta)=&\arg\underset{\mathbf{w}}{\min}&\quad tr(\mathbf{x}\mathbf{x}^H)\\\nonumber
&s.t.&\angle(\mathbf{h}_j\mathbf{x})=\angle(d_j)\quad\forall j\in K\\
&\quad&\mathbf{h}_j\mathbf{x}\mathbf{x}^H\mathbf{h}^H_j=\zeta_j\quad
\forall j\in K.
\end{eqnarray}
is given by $\mathbf{x}_{e}(d,\mathbf{A}(d,d_j)\mathbf{H})$ in (\ref{powcd}), where $\mathbf{A}(d,d_j)$
\vspace{-0.2cm}
\begin{eqnarray}
\mathbf{A}(d,d_j)=\begin{cases}
\exp((\angle d-\angle d_j)i), \quad j=k\\
0,\quad j\neq k.  
\end{cases}
\end{eqnarray}
\end{thm1}
\end{newtheorem}
\begin{proof} We assume that we have the following \textit{equivalent} channel as
\vspace{-0.3cm}
\begin{eqnarray}
\mathbf{H}_e=\mathbf{A}\mathbf{H}.
\end{eqnarray}
The power minimization can be rewritten by replacing $\mathbf{H}$ by its
equivalent channel $\mathbf{H}_e$ in (\ref{powcd}) as
\vspace{-0.1cm}
 \begin{eqnarray}\nonumber
\label{pow1}
&\underset{\mathbf{x}_e}{\min}&\quad (\mathbf{x}^H_e\mathbf{x}_e)\\\nonumber
&s.t.&\angle(\mathbf{h}_{e,j}\mathbf{x}_e)=\angle(d)\quad\forall j\in K\\
&\quad&\mathbf{h}_{e,j}\mathbf{x}_e\mathbf{x}^H_e\mathbf{h}^H_{e,j}=\zeta_j\quad
\forall j\in K.
\end{eqnarray} 
where $\mathbf{h}_{e,j}$ is the $j^{th}$ row of the $\mathbf{H}_e$. 
Rewriting the first constraints in (\ref{pow1}) as 
\vspace{-0.2cm}
\begin{eqnarray}\nonumber
&\quad&\angle(d-d_j)\angle(\mathbf{h}_{j}\mathbf{x}_e)=\angle(d)\\
&\equiv&\angle(\mathbf{h}_{j}\mathbf{x}_e)=\angle(d_j) \quad\forall j\in K
\end{eqnarray}
shows the equivalence  between the constrained constellation  multicast
channel and
constructive interference downlink channel. 
\end{proof}

By taking a look at (\ref{unitrank}), the objective function $\|\mathbf{x}\|^2$
is unit rank and thereby it is a convex. The convexity holds for $\mathcal{C}_2$,
however, the phase constraints in $\mathcal{C}_1$ are not convex. Therefore
a formulation for $\mathcal{C}_1$ is required.  
 We can reformulate the constraint as
 \vspace{-0.1cm}
\begin{eqnarray}
\label{powal1}
&\hspace{-2.9cm}\quad&\mathbf{x}(d_j,\mathbf{H},\boldsymbol\zeta)=
\arg\underset{\mathbf{x}}{\min}\quad \|\mathbf{x}\|^2\\\nonumber
&s.t.&\begin{cases}\mathcal{C}1:\frac{\mathbf{h}_j\mathbf{x}-(\mathbf{h}_j\mathbf{x})^H}{i({\mathbf{h}_j\mathbf{x}+(\mathbf{h}_j\mathbf{x})^H})}=\tan(d),
\forall
j\in K\\
\mathcal{C}2:\mathcal{R}\{d_j\}.\mathcal{R}\{\mathbf{h}_j\mathbf{x}\}\geq 0, \forall
j\in K\\
\mathcal{C}3:\mathcal{I}\{d_j\}.\mathcal{I}\{\mathbf{h}_j\mathbf{x}\}\geq 0, \forall
j\in K\\
\mathcal{C}4:\|\mathbf{h}_j\mathbf{x}\|^2\geq \sigma^2\zeta_j,\forall j\in K.
\end{cases}
\end{eqnarray}
 The minimum transmit power in (\ref{powccm})-(\ref{unitrank})
 occurs when the inequality constraints are replaced by equality (i.e. all
 users should achieve their target threshold SNR). A final formulation can be expressed as 
\vspace{-0.1cm} 
\begin{eqnarray}
\label{powa}
\hspace{-0.9cm}&\quad&\arg\underset{\mathbf{x}}{\min}\quad \|\mathbf{x}\|^2\\\nonumber
&s.t.&\begin{cases}\mathcal{C}1:\frac{\mathbf{h}_j\mathbf{x}-(\mathbf{h}_j\mathbf{x})^H}{2i}=\sigma\sqrt{\zeta_j}\mathcal{I}\{d\},
\forall
j\in K\\
\mathcal{C}2:\frac{\mathbf{h}_j\mathbf{x}+(\mathbf{h}_j\mathbf{x})^H}{2}=\sigma\sqrt{\zeta_j}\mathcal{R}\{d_j\},\forall
j\in K.\end{cases}
\end{eqnarray}
It can be viewed that the constraints in (\ref{powal1}) are turned from inequality
constraints to equality constraint (\ref{powa}) due to signal aligning requirements. The Lagrangian function can be derived as follows
\vspace{-0.2cm}
\begin{eqnarray}\nonumber
\hspace{-1.9cm}&\mathcal{L}&(\mathbf{x})=\|\mathbf{x}\|^2\\\nonumber
&+&\sum_j{\mu_j}\bigg(-0.5i\small(\mathbf{h}_j\mathbf{x}-\mathbf{x}^H\mathbf{h}^H_j\small)-\sqrt{\zeta_{j}}\mathcal{I}\{d_j\}\bigg)\\\nonumber
&+&\sum_j{\alpha_j}\bigg(0.5\small(\mathbf{h}_j\mathbf{x}+\mathbf{x}^H\mathbf{h}^H_j\small)
-\sqrt{\zeta_{j}}\mathcal{R}\{d_j\}\bigg)\\
\end{eqnarray}
where $\mu_j$ and $\alpha_j$ are the Lagrangian dual variables. The derivative for the Lagrangian function can be written as
\vspace{-0.1cm}
\begin{eqnarray}
\frac{d\mathcal{L}(\mathbf{x})}{d\mathbf{x}^*}=\mathbf{x}+0.5i\sum_j\mu_j\mathbf{h}^H_j
+0.5\sum_j\alpha_j\mathbf{h}^H_j
\end{eqnarray}
\vspace{-0.15cm}
By equating this term to zero, $\mathbf{x}_i$ can be written as
\hspace{-0.2cm}\begin{eqnarray}\nonumber
\label{CIPM}
\hspace{-0.5cm}\mathbf{x}&=&-0.5i\sum^K_{j=1}\mu_j\mathbf{h}^H_j
-0.5\sum_j\alpha_j\mathbf{h}^H_j\\
 &\equiv&\sum^K_{j=1}\nu_j\mathbf{h}^H_j,\forall
i\in K
\end{eqnarray}
where $\nu_j\in \mathbb{C}=-0.5i\mu_j-0.5\alpha_j$. The optimal values of the Lagrangian variables $\mu_j$ and $\alpha_j$ can
be found by substituting $\mathbf{w}$ in the constraints (\ref{powa}) which result in solving the set of $2K$ equations (\ref{setoo}).
 The final constrained 
constellation multicast
precoder can be found by substituting all $\mu_j$ and $\alpha_j$ in (\ref{CIPM}). 
 
\begin{figure*}[t]
\vspace{-0.2cm}
\hspace{0.2cm}
\begin{tabular}[t]{c}
\begin{minipage}{17 cm}
 \begin{eqnarray}
\label{multicasteq}
\begin{array}{cccc}
\label{setoo}
0.5K\|\mathbf{h}_1\|(\sum_k(-\mu_k+\alpha_ki)\|\mathbf{h}_k\|\rho_{1k}&-&\sum_k(-\mu_k+\alpha_ki)\|\mathbf{h}_k\|\rho^{*}_{1k})=\sqrt{\zeta^{}_{1}}\mathcal{I}(d_1)\\
0.5K\|\mathbf{h}_1\|(\sum_k(-\mu_ki-\alpha_k)\|\mathbf{h}_k\|\rho_{1k}&+&\sum_k(-\mu_ki-\alpha_k)\|\mathbf{h}_k\|\rho^{*}_{1k})=\sqrt{\zeta^{}_{1}}\mathcal{R}(d_1)\\
\quad&\vdots&\\
0.5K\|\mathbf{h}_K\|(\sum_k(-\mu_k+\alpha_ki)\|\mathbf{h}_k\|\rho_{Kk}&-&\sum_k(-\mu_k+\alpha_ki)\|\mathbf{h}_k\|\rho^{*}_{Kk})=\sqrt{\zeta_{K}}\mathcal{I}(d_K)\\
0.5K\|\mathbf{h}_K\|(\sum_k(-\mu_ki-\alpha_k)\|\mathbf{h}_k\|\rho_{Kk}&+&\sum_k(-\mu_ki-\alpha_k)\|\mathbf{h}_k\|\rho^{*}_{Kk})=\sqrt{\zeta_{K}}\mathcal{R}(d_K)\\
\end{array}
\end{eqnarray}
\end{minipage}\\
\hline
\hline
\end{tabular}
\end{figure*}
\begin{cor}
The   CI precoding for power minimization  $\mathbf{x}_{CIPM}$ as well as
constrained constellation multicast precoding  must span the subspaces of each user's channel.\smallskip
\end{cor}

It can be noted from the formulation of $\mathbf{x}_{CIPM}$ that BS should use
the same precoder for all users. This result resembles the multicast approach
in which the BS wants to deliver the same message to all users\cite{multicast}-\cite{multicast-jindal}.
However in multicast systems, a different symbol should be detected correctly
at each user.
\smallskip

Using (\ref{CIPM}), we can rewrite the received signal at $j^{th}$ receiver
as 
\vspace{-0.1cm} 
\begin{eqnarray}\nonumber
\hspace{-05cm}\vspace{-0.2cm}
\label{alignment} 
y_j&=&\mathbf{h}_j\mathbf{x}+z_j=\mathbf{h}_j\sum^K_{k=1}\nu_k\mathbf{h}^H_k+z_j\\\vspace{-0.3cm}\nonumber
&\label{cm1}\equiv&\mathbf{h}_j\begin{bmatrix}|\nu_1|*\mathbf{h}^H_1\quad \hdots\quad|\nu_K|*\mathbf{h}^H_K\end{bmatrix}
\begin{bmatrix}d*1\angle(\nu_1)\\
\vdots\\\nonumber
d*1\angle(\nu_K)\end{bmatrix}+z_j.\\\nonumber
\end{eqnarray}
From (\ref{cm1}), the constellation constrained multicast can be formulated as a constructive interference downlink channel with set of precoders $\mathbf{h}^H_1,
\hdots,\mathbf{h}^H_K$, where each one of these precoders is allocated with power $|\nu_k|$ and associated with the symbol $d*1\angle \nu_k$.\smallskip

\begin{cor}
\label{eq}
The solution of problem $\mathbf{x}_{CIPM}$ with uniformly scaled SINR constraints is given simply by scaling the output vector of the original problem as follows:
\begin{eqnarray}\nonumber
\mathbf{x}_{CIPM}(\mathbf{d},\mathbf{H},n\boldsymbol\zeta)=\sqrt{n}\mathbf{x}_{CIPM}(\mathbf{d},\mathbf{H},\boldsymbol\zeta)
\end{eqnarray}
where $n\in \mathbb{R}^+$.\smallskip   
\end{cor}
\begin{proof}
We define the normalized precoder $\mathbf{\hat{x}}_{CIPM}$ equals to $\frac{\mathbf{x}_{CIPM}}{\|\mathbf{x}_{CIPM}\|}$.
For any $\mathbf{x}_{CIPM}$, $\angle(\mathbf{h}_k\mathbf{x}_{CIPM})=\angle(\mathbf{h}_k\mathbf{\hat{x}}_{CIPM}),
\forall k\in K$. Therefore, all users can receive their target data symbols
$\mathbf{d}$ scaled to a certain SNR value. This implies that scaling uniformly all
users' SNR targets does change $\mathbf{\hat{x}}_{CIPM}$. Using
the simultaneous set of equations (\ref{setoo}), we can replace each $\zeta_j$
by $n\zeta_j$. This multiplies each value of $\mu_j$, $\alpha_j$ by
$\sqrt{n}$. As a consequence, a scaling factor of $\sqrt{n}$ is multiplied
with the original output vector $\mathbf{x}$ which proves the corollary.
\end{proof}
\subsection{Constructive Interference Power minimization bounds}
In order to assess the performance of the proposed algorithm, we mention
two theoretical upper bound as follows

\subsubsection{Genie aided upper bound}
This bound occurs when all multiuser transmissions are constructively
interfering by nature and without the need to optimize the output vector. The minimum transmit power for a system that exploits the constructive interference
on symbol basis can be found by the following theoretical bound\smallskip
 
\begin{thm1}
The genie-aided minimum transmit power in the downlink of multiuser MISO   system can be found by solving the following optimization
\begin{eqnarray}\nonumber
\label{pr}
\hspace{-0.1cm}{P}_{min}&=&\arg\underset{p_1,\hdots,p_K}{\min}\quad\sum^K_{k=1}p_k\\\nonumber
&s.t.&\|\mathbf{g}_k\|^2(|\xi_{kk}|^2{p_k}+\sum^K_{j=1,j\neq k}{p_j}|\xi_{kj}|^2)\geq{\zeta_k},\forall
k\in K.\\
\end{eqnarray}
\end{thm1}

\begin{proof}According to (\ref{rot}), the bound in (\ref{pr}) can be found if all users face a constructive interference
with respect to the multiuser transmissions of all other streams $\angle(\xi_{jk}d_j)=\angle
d_k,\forall k,\forall j$.
\end{proof}
This bound can be mathematically found by solving the problem (\ref{pr}) using linear programming
techniques\cite{boyd}. \\
\subsubsection{Optimal Multicast} 
Based on theorem (2), a theoretical upperbound can be characterized. This bound occurs if we drop the phase alignment
constraint $\mathcal{C}_1$. The intuition of
using this technique is the complete correlation among the information that needs to be communicated (i.e. same symbol for all users). The optimal input covariance for power minimization in multicast system can
be found as a solution of the following optimization
\vspace{-0.05cm} 
\begin{eqnarray}
\label{powm1}
&\underset{\mathbf{Q}:\mathbf{Q}\succeq 0}{\min}&\quad tr(\mathbf{Q})\quad s.t.\quad\mathbf{h}_j\mathbf{Q}\mathbf{h}^H_j\geq\zeta_j\quad,\forall j\in K.
\end{eqnarray}
This problem is thoroughly solved in \cite{multicast}. A tighter upperbound
can be found by imposing a unit rank constraint on $\mathbf{Q}$\cite{multicast-jindal},
to allow the comparison with the unit rank transmit power minimization constructive interference
precoding 
\begin{eqnarray}
\label{unitrankm}
\underset{\mathbf{Q}:\mathbf{Q}\succeq 0, \text{Rank}(\mathbf{Q})=1}{\min} tr(\mathbf{Q})\quad s.t.\quad\mathbf{h}_j\mathbf{Q}\mathbf{h}^H_j\geq\zeta_j\quad,\forall j\in K
\end{eqnarray}
Eq. (\ref{unitrankm}) presents a tighter upper bound in comparison (\ref{powm1}).
It assumes a unit rank approximation of (\ref{powm1}). 
\section{Weighted Max Min SINR Algorithm for Constructive Interference Precoding (CIMM)}\label{maxmins}
The weighted max-min SINR beamforming aims at improving the relative fairness in the system by maximizing
the worst user SINR. This problem has been studied in different frameworks
 such as multicast \cite{multicast}, and downlink transmissions\cite{shitz}. In
 \cite{multicast}, the authors have solved this problem by finding the relation
 between the min-pwr problem and max-min problem and formulating both problem
 as convex optimization ones. On the other hand, the authors of \cite{shitz} have solved the problem using the bisection technique. In this
 work, we exploit the constructive interference to enhance the user fairness in terms of weighted SNR. The challenging aspect is the additional constraints which guarantee the data have been detected correctly at
 the receivers. The constructive interference max-min problem can be formulated as
\begin{eqnarray}
\label{maxmin}
\mathbf{w}_k=&\underset{\mathbf{w}_k}{\max}\underset{j}{\min}&\Big\{\frac{\|\mathbf{h}_j\sum^K_{k=1}\mathbf{w}_kd_k\|^2}{r_j}\Big\}^K_{i=1}\\\nonumber
&{s.t.}&\begin{cases}\mathcal{C}1:\|\sum^K_{k=1}\mathbf{w}_kd_k\|^2\leq P\\\nonumber
\mathcal{C}2:\angle(\mathbf{h}_j\sum^K_{k=1}\mathbf{w}_kd_k)=\angle(d_j),\quad\forall j\in K.\\
\end{cases}\end{eqnarray}
where $r_i$ denotes the requested SNR target for the $i^{th}$ user. If we
denote $\mathbf{q}=\sum^K_{j=1}\mathbf{w}_jd_j$, the previous optimization
can be formulated as 
\begin{eqnarray}
\label{maxminq}
\mathbf{q}(\mathbf{d},\mathbf{H},\mathbf{r})=&\underset{\mathbf{q}}{\max}\underset{j}{\min}&\Big\{\frac{\|\mathbf{h}_j\mathbf{q}\|^2}{r_j}\Big\}^K_{i=1}\\\nonumber
&{s.t.}&\begin{cases}\mathcal{C}1:\|\mathbf{q}\|^2\leq P\\\nonumber
\mathcal{C}2:\angle(\mathbf{h}_j\mathbf{q})=\angle(d_j),\quad\forall j\in K\\
\end{cases}\end{eqnarray}
where $\mathbf{r}$ is the vector that contains all the weights $r_i$. In the following, it is shown
that the optimal output vector is a scaled version of the min-pwr solution in (\ref{powccm})\cite{multicast}. The weighted maximum minimum SINR problem has been solved using bisection method over  $t\in[0,1]$\cite{shitz}. \vspace{-0.5cm}
\subsection{Max-min SINR and min-pwr relation}
\begin{lemma} The relationship between min power and max-min problem can
be described  as $\mathbf{q}(\mathbf{d},\mathbf{H},\mathbf{r})=\mathbf{x}(\mathbf{d},\mathbf{H},{t}^{*}\mathbf{r})$.\smallskip
\begin{proof}
The problem (\ref{maxminq}) can be formulated
 \hspace{-0.4cm}\begin{eqnarray}
&\underset{t,\mathbf{q}}{\max}&\quad t\\\nonumber
&s.t.&\begin{cases}\mathcal{C}1:\|\mathbf{q}\|^2\leq P\\
\mathcal{C}2:\frac{\mathbf{h}_j\mathbf{q}-(\mathbf{h}_j\mathbf{q})^H}{i({\mathbf{h}_j\mathbf{q}+(\mathbf{h}_j\mathbf{q})^H})}=\tan(\angle
d_j),
\forall
j\in K\\
\mathcal{C}3:\mathcal{R}\{d_j\}.\mathcal{R}\{\mathbf{h}_j\mathbf{q}\}\geq 0, \forall
j\in K\\
\mathcal{C}4:\mathcal{I}\{d_j\}.\mathcal{I}\{\mathbf{h}_j\mathbf{q}\}\geq 0, \forall
j\in K\\
\mathcal{C}5:\|\mathbf{h}_j\mathbf{q}\|^2\geq R_jt,\forall j\in K.
\end{cases}
\end{eqnarray}
The optimal value of $t$  denoted by $t^{*}$ can be found by solving the min-pwr.
\end{proof}
\end{lemma}

Thus, the max-min SINR solution is a scaled version of min power
solution, which means that the system designer needs to find the optimal
value of $t^*$ to solve the max-min problem. In the next section, we propose
a simple method that can find this parameter influenced by the literature
\cite{shitz}.
\vspace{-0.5cm}
\subsection{Max-min SINR Constructive Interference Precoding}

In comparison with (\ref{maxmin}), we have additional $3K$ constraints that
limit the system performance. The problem can be formulated as
\begin{eqnarray}
\label{cmmaxmin}
\hspace{-0.1cm}&\underset{t,\mathbf{q}}{\max}&\quad t\\\nonumber
&s.t.&\begin{cases}\mathcal{C}1:\|\mathbf{q}\|^2= P\\
\mathcal{C}2:\frac{\mathbf{h}_j{\mathbf{q}}-(\mathbf{h}_j{\mathbf{q})}^H}{{i({\mathbf{h}_j{\mathbf{q}}+(\mathbf{h}_j{\mathbf{q}})^H})}}=\tan(\angle
d_j),
\forall
j\in K\\
\mathcal{C}3:\frac{(d+d^*)}{2}.\frac{\mathbf{h}_j\mathbf{q}+(\mathbf{h}_j\mathbf{q})^H}{2}\geq 0, \forall
j\in K\\
\mathcal{C}4:\frac{(d-d^*)}{2i}.\frac{\mathbf{h}_j\mathbf{q}-(\mathbf{h}_j\mathbf{q})^H}{2i}\geq 0, \forall
j\in K\\
\mathcal{C}5:\|\mathbf{h}_j\mathbf{q}\|^2\geq r_jt,\forall j\in K.
\end{cases}
\end{eqnarray}\smallskip

A solution for (\ref{cmmaxmin}) can be found in the same fashion by using
the bisection method as \cite{shitz} and can be summarized as\\

\smallskip
\begin{center}
\hspace{-0.5cm}\begin{tabular}{p{8.3cm}}
\hline
\textbf{A2}: Bisection for max-min SINR for CI precoding (CIMM)\\
\hline
\hspace{0.2cm} $m_1\rightarrow 0$\\
\hspace{0.2cm} $m_2\rightarrow 1$\\
\hspace{0.2cm} Repeat\\
\hspace{0.2cm} set $t_m=\frac{m_1+m_2}{2}$\\
\hspace{0.2cm}$\mathbf{q}(t_m)=\text{find}\quad\mathbf{x}_{CIPM}(\mathbf{d},\mathbf{H}, t_m\mathbf{r})$\
\hspace{0.1cm}set $\hat{P}=\|\mathbf{x}_{CIPM}\|^2$\\
\hspace{0.2cm}if $\hat{P}\leq P$\\
\hspace{0.2cm}then\quad $t_1 \rightarrow t_m$\\
\hspace{0.2cm}else\quad$t_2 \rightarrow t_m$\\
Until $|\hat{P}-P|\leq \delta$\\
Return $t_m$\\
\hline\\
\end{tabular}
\end{center}
\section{Weighted Sum Rate Maximization Algorithms for constructive interference
Precoding (CISR)}\label{wsr}
The sum rate problem of the multiuser downlink of multiple antennas for user-level
precoding has been investigated in the literature \cite{coordinated}-\cite{bjornson}. The authors in \cite{coordinated} prove that the sum rate problem is NP
hard. However, a simpler
solution for the sum rate problem is characterized in \cite{bjornson} by
rotating the MRTs of each user's channel to reduce the amount of the created
interference on other users' transmissions. On the other hand,
the weighted sum rate optimization in single group multicast scenarios is studied
\cite{multicastc}, which tries to design closed form precoders at different
high SNR scenarios and proposes an iterative algorithm with low computation complexity for general SNR case. Furthermore, heuristic solutions for sum
rate maximization of group
multicast precoding with per-antenna power constraint are proposed
in \cite{dimitris_multi}.

In this work, we take into the account that the interference  can be exploited
among the different multiuser data streams. This requires that the sum rate
problem should be formulated to take into consideration this new feature. The weighted sum rate maximization with a unit rank assumption for the precoding
matrix can be written as\footnote{For the sum rate problrm, it should be noted that the optimal solution is not necessarily unit rank, but we employ this assumption to enable tractable heuristic solutions}
\vspace{-0.1cm}
\begin{eqnarray}\nonumber
\label{SR}
&\underset{\mathbf{q}}{\max}&\quad\sum^K_{j=1}\phi_j\log_2(1+\frac{\|\mathbf{h}_j\mathbf{q}\|^2}{\sigma^2})\\
&s.t.&\begin{cases}\mathcal{C}_1:\angle\mathbf{h}_j\mathbf{q}=\angle d_j \quad\forall j\in K.\\
\mathcal{C}_2:\|\mathbf{q}\|^2\leq P.
\end{cases}\end{eqnarray}
where $\phi_j$ is the weight related to the $j^{th}$ user.
The optimization can be
formulated as 
\begin{eqnarray}
\label{sumrate}
\hspace{-0.9cm}&\underset{\mathbf{q}}{\max}&\quad\sum^K_{j=1}\phi_j\log_2(1+\frac{\|\mathbf{h}_j\mathbf{q}\|^2}{\sigma^2})\\\nonumber
&s.t.&\begin{cases}\mathcal{C}_1:\mathcal{I}\{\mathbf{h}_j\mathbf{q}\}=i\tan(\angle d_j)\mathcal{R}\{\mathbf{h}_j\mathbf{q}\} \quad\forall j\in
K,\\\nonumber
\mathcal{C}_2:\|\mathbf{q}\|^2\leq P,\\\nonumber
\mathcal{C}_3:\frac{(d_j+d_j^*)}{2}.\frac{\mathbf{h}_j\mathbf{q}+(\mathbf{h}_j\mathbf{q})^H}{2}\geq
0, \forall j\in K,\\\nonumber
\mathcal{C}_4:\frac{(d_j-d_j^*)}{2i}.\frac{\mathbf{h}_j\mathbf{q}-(\mathbf{h}_j\mathbf{q})^H}{2i}\geq 0, \forall j\in K.\\
\end{cases}
\end{eqnarray}
\subsection{Modulation Selection}

In order to optimize the sum-rate, practical communication systems implement
adaptive modulation and coding schemes (MCS) which adapt the density of the transmitted constellation to the current SNR. Unfortunately, this adaptation cannot be  applied on a symbol-level because this would render the signalling overhead impractical. In this context, let us assume that the modulation of each user remains fixed during the channel coherence time of a quasi-static block fading channel. This way, each user has to be notified only once per $\tau_c$ about the constellation type that he has to detect. 
 The modulation for each user is selected at the beginning of each coherence
 time. To decide the most appropriate modulation
 type for each user, we use the optimal multicast, which is defined as
 \begin{eqnarray}
\label{sum_rate_mu}
\mathbf{Q}_{o}=\arg\underset{\mathbf{Q}:tr(\mathbf{Q})\leq P}{\max}\sum^K_{j=1}\phi_j\log_2(1+\frac{\mathbf{h}_j\mathbf{Q}\mathbf{h}^H_j}{\sigma^2}),
\end{eqnarray} 
to decide the highest modulation
 order for each user for the whole transmission frame (i.e. here is assumed
 to be equal to the channel coherence time) by the following criteria 
\begin{eqnarray}
\label{modcode}
\text{MCS}=\begin{cases}\text{no service}, \zeta^{'}_j\leq \zeta_{BPSK}\\
\text{BPSK},\zeta_{BPSK}\leq\zeta^{'}_j \leq\zeta_{QPSK}\\
\vdots\\
\\
\end{cases}
\end{eqnarray}
where $\zeta^{'}_j$ is the effective SNR at each user receiver.
Due to the fact that M-PSK data symbols are encoded
using the phase information, any higher order PSK symbol can be decoded
as lower PSK if each user knows its target modulation. For
example, the symbol $1\angle 45^\circ=\frac{1+i}{2}$
can be detected as $11$ if the agreed modulation between the transmitter and
receiver is QPSK, and it can be be detected as $1$ if the agreed modulation
is BPSK. As a result, when designing sum-rate CI algorithms at the transmitter, we can always assume that all users expect the highest-order PSK modulation. At the receiver, the demapping of the received symbol will depend on the assigned MCS.

\vspace{-0.2cm}
\subsection{Genie aided sum rate upper bound}
Based on  theorem (\ref{thmo}), the maximum sum rate can be found by solving
the following optimization
\hspace{-0.5cm}\begin{eqnarray}\nonumber
&\hspace{-0.5cm}\underset{p_1,\hdots,p_K}{\max}&\quad\sum^K_{k=1}\log_2\big(1+\|\mathbf{g}_k\|^2(|\xi_{kk}|^2{p_k}+\sum^K_{j=1,j\neq k}{p_j}|\xi_{kj}|^2)\big)\\
&\hspace{-0.5cm}s.t.&\sum^K_{k=1}p_k\leq P
\end{eqnarray}

\vspace{-0.4cm}
\subsection{Optimal solution}
The optimal solution cannot be found in a straight forward manner due to
the
different types of the constraints: the phase constraints and the threshold
constraints. However, we write the
Lagrangian function of the previous optimization (\ref{sumrate}) to get more insights about the problem
as (\ref{sum_rate}) and the derivative of the related sum rate problem (\ref{lagrange_sumrate}). Moreover, it
can be seen that it has different solutions than the power minimization problem. \\
\begin{figure*}[t]
\hspace{0.2cm}
\begin{tabular}[t]{c}
\begin{minipage}{17 cm}
\begin{eqnarray}\nonumber
\label{sum_rate}
\mathcal{L}(\mathbf{q})&=&\sum^K_{j=1}\phi_j\log_2(1+\frac{\|\mathbf{h}_j\mathbf{q}\|^2}{\sigma^2})+\sum^K_{j=1}\mu_j\Big(\frac{(d_j+d_j^*)}{2}\frac{\mathbf{h}_j\mathbf{q}+(\mathbf{h}_j\mathbf{q})^H}{2}\Big)+\sum^K_{j=1}\alpha_j\Big(\frac{(d_j-d_j^*)}{2i}\frac{\mathbf{h}_j\mathbf{q}-(\mathbf{h}_j\mathbf{q})^H}{2i}\Big)\\
&+&\sum^K_{j=1}\kappa_j\Big(\mathbf{h}_j\mathbf{q}-\mathbf{q}^H\mathbf{h}^H_j-\tan(\angle
d_j)\big(\mathbf{h}_j\mathbf{q}+\mathbf{q}^H\mathbf{h}^H_j\big)\Big)+\gamma(\mathbf{q}^H\mathbf{q}-P)
\end{eqnarray}
\vspace{-0.2cm}
\begin{eqnarray}
\label{lagrange_sumrate}
\frac{d\mathcal{L}(\mathbf{q})}{d\mathbf{q}^*}=\sum^K_{j=1}\phi_j\frac{\mathbf{h}^H_j\mathbf{h}_j\mathbf{q}}{\sigma^2+\mathbf{q}^H\mathbf{h}^H_j\mathbf{h}_j\mathbf{q}}+\sum^K_{j=1}\mu_j\frac{d_j+d^*_j}{4}\mathbf{h}^H_j+\sum^K_{j=1}\alpha_j\frac{(d_j-d_j^*)}{2}\mathbf{h}^H_j-\sum^K_{j=1}\kappa_j(1+\tan(\angle
d_j))\mathbf{h}^H_j+\gamma\mathbf{q}
\end{eqnarray}
\end{minipage}\\
\hline
\hline
\end{tabular}
\end{figure*}
\vspace{-0.2cm}
\subsubsection{Low SNR approximation} To simplify the analysis, we use the low SINR approximation
$\log_2(1+\frac{\alpha_i}{\sigma^2})\sim \frac{\alpha_i}{\sigma^2}$ which
is valid in the regime $\sigma^2\rightarrow \infty$. Thus
the optimization problem (\ref{sumrate}) can be written as
\begin{eqnarray}
\label{L_sum0}
\hspace{-0.2cm}&\underset{\mathbf{q}}{\max}&\quad\sum^K_{j=1}\phi_j\frac{\|\mathbf{h}_j\mathbf{q}\|^2}{\sigma^2}\\\nonumber
&s.t.&\begin{cases}\mathbf{h}_j\mathbf{q}-\mathbf{q}^H\mathbf{h}^H_j=i\tan(\angle d_j)(\mathbf{h}_j\mathbf{q}+\mathbf{q}^H\mathbf{h}^H_j)\quad\forall j\in K,\\\nonumber
{\|\mathbf{q}\|^2}\leq P.\\
\mathcal{C}_3:\frac{(d_j+d_j^*)}{2}\frac{\mathbf{h}_j\mathbf{q}+(\mathbf{h}_j\mathbf{q})^H}{2}\geq
0, \forall j\in K,\\\nonumber
\mathcal{C}_4:\frac{(d_j-d_j^*)}{2i}\frac{\mathbf{h}_j\mathbf{q}-(\mathbf{h}_j\mathbf{q})^H}{2i}\geq 0, \forall j\in K\end{cases}
\end{eqnarray}
and the corresponding Lagrangian function can be written as
\begin{eqnarray}\nonumber
\label{L_sum}
\hspace{-0.1cm}\mathcal{L}(\mathbf{q})&=&\sum^K_{j=1}\phi_j\frac{\|\mathbf{h}_j\mathbf{q}\|^2}{\sigma^2}+\sum^K_{j=1}\pi_j\Big(\mathbf{h}_j\mathbf{q}-\mathbf{q}^H\mathbf{h}^H_j-i\tan(\angle d_j)\\\nonumber
&\times&(\mathbf{h}_j\mathbf{q}+\mathbf{q}^H\mathbf{h}^H_j)\Big)+\sum^K_{j=1}\beta_j\frac{(d_j+d_j^*)}{2}\frac{\mathbf{h}_j\mathbf{q}+(\mathbf{h}_j\mathbf{q})^H}{2}\\\nonumber
&+&\sum^K_{j=1}\alpha_j\frac{(d_j-d_j^*)}{2i}\frac{\mathbf{h}_j\mathbf{q}-(\mathbf{h}_j\mathbf{q})^H}{2i}+\zeta\Big(\mathbf{q}^H\mathbf{q}-P\Big).\\
\end{eqnarray}
If we denote $\alpha^{'}_j=\alpha_j\frac{d_j-d^*_j}{2i}$, $\beta^{'}_j=\beta_j\frac{d_j+d^*_j}{2}$,
$\pi^{'}_j=-\pi_j(1+\tan(\angle d_j))$, the derivative of (\ref{L_sum}) can be
formulated as
\begin{eqnarray}
\frac{d\mathcal{L}(\mathbf{q})}{d\mathbf{q}^*}=\sum^K_{j=1}\phi_j\frac{\mathbf{h}^H_j\mathbf{h}_j\mathbf{q}}{\sigma^2}+\sum^K_{j=1}\Big(\alpha^{'}_{j}+\pi^{'}_j+\beta^{'}_j\Big)\mathbf{h}^H_j+\zeta\mathbf{q}.
\end{eqnarray}
Then, $\mathbf{q}$ can be expressed as 
\begin{eqnarray}
\label{lowsnr}
\mathbf{q}=\Big(\sum^K_{j=1}\phi_j\frac{\mathbf{h}^H_j\mathbf{h}_j}{\sigma^2}+\zeta\mathbf{I}\Big)^{-1}\bigg(\sum^K_{j=1}\Big(\alpha^{'}_{j}+\pi^{'}_j+\beta^{'}_j\Big)\mathbf{h}^H_j\bigg).
\end{eqnarray}
 Since the assumed approximation works in the low SNR regime (i.e. noise limited scenario $\sigma^2\rightarrow \infty$),
 the expression in (\ref{lowsnr}) can be simplified into the following expression
 \begin{eqnarray}
 \label{sumrate_w}
 \mathbf{q}=\sum^K_{j=1}\underset{{a_j}}{\underbrace{\Big(\pi^{'}_j+\beta^{'}_j+\alpha^{'}_j\Big)}}\mathbf{h}^H_j
.
 \end{eqnarray}
 It can be noted that the precoding formulation at the low SNR regime resembles the generic
 formula of pwr-min precoding. Moreover, the weight for each user vanishes in this regime.
 Based on this fact, we propose heuristic precoding schemes that aim at maximizing
 the sum rate of the downlink multiuser transmissions.
\vspace{-0.3cm}
\subsection{Heuristic schemes }
\vspace{-0.1cm}
Since the solution for the sum rate maximization problem in (\ref{SR}) is difficult to find, we propose two heuristic algorithms to tackle this problem as follows  \smallskip
\subsubsection{Phase
alignment algorithm}
The sum rate maximization problem can be solved exploiting the low SNR approximation expression in (\ref{sumrate_w}). This expression contains $3K$ variables
(i.e. $\alpha^{'}_j$, $\beta^{'}_j$, and $\pi^{'}_j$ $\forall j\in K$) that have to satisfy the phase alignment constraints $\mathcal{C}_1$ in (\ref{SR})
while it should maximize the sum rate in the system. Utilizing the eigenvectors
of $\mathbf{H}\mathbf{H}^H$, $\mathbf{q}$ can be formulated as
\vspace{-0.25cm} 
\begin{eqnarray}\nonumber
\mathbf{q}=\sum^K_{j=1}a_j\mathbf{h}^H_j=\sum^K_{j=1}b_j\mathbf{e}_j,
\end{eqnarray}
where $\mathbf{e}_j$ is the j$^{th}$ eigenvector of $\mathbf{H}\mathbf{H}^H$.
This makes the received SINR formulated as
\begin{eqnarray}
\zeta_{i}=\sum^K_{j=1}|b_j|^2\mathbf{h}_i\mathbf{e}_j\mathbf{e}^H_j\mathbf{h}^H_i
\end{eqnarray} 
The optimization (\ref{sumrate}) can be reformulated as
\begin{eqnarray}\nonumber
\label{eigenSR}
&\underset{{b_i}}{\max}&\quad\sum^K_{j=1}\log_2(1+\zeta_j)\\
&s.t.&\quad\begin{cases}\mathcal{C}_1:\sum^K_{i=1}|{b}_i|^2\leq P.\\
\mathcal{C}_2:\angle\mathbf{h}_j(\sum^K_{i=1}b_i\mathbf{e}_i)=\angle d_j, ~\forall j\in
K\end{cases}
\end{eqnarray}
\footnotesize
\begin{tabular}{p{8.3cm}}
\hline
\textbf{A3}: Sum Rate Maximization - Phase alignment algorithm (CISR-PA)\\
\hline
\begin{enumerate}
\item Solve the optimization (\ref{eigenSR}) without $\mathcal{C}_2$ and
find $|b_i|$.
\item Select the modulation type for each user based on the achieved SINR
$\zeta_j$.
\item Solve the following set of equations by finding $\angle b_i$
\begin{eqnarray}
\mathbf{h}_j(\sum^K_{i=1}|b_i|\mathbf{e}_i\exp(\angle b_i))=\zeta_j,\forall
j\in K.
\end{eqnarray}
\item Scale $\mathbf{q}$ by setting $\|\mathbf{q}\|^2=P$

\end{enumerate}\\
\hline

\end{tabular}\\
 \normalsize
  
\subsubsection{ Greedy Algorithm}

In this algorithm, we jointly utilize the solution for the unconstrained optimal sum rate maximization multicast problem (\ref{sum_rate_mu})
and the constrained constellation power minimization problem (\ref{powccm}) to
propose a new heuristic algorithm. 
The intuition of using such algorithm is to find the subset
of users that has similar characteristics in terms of the co-linearity
with respect to the optimal multicast directions (i.e. the projection to
eigenvector associated with the maximum eigenvalue of $\mathbf{Q}_{o}$). It should be noted that this co-linearity is defined by the projection
and the angle of the projection as it is illustrated in the following algorithm\\

\footnotesize
\begin{center}
\begin{tabular}{p{8.3cm}}
\hline
\textbf{A4}: Greedy sum rate maximization\quad (CISR-G)\\
\hline
\begin{enumerate}
\item Find the optimal input covariance by solving the unconstrained multicast
problem (\ref{sum_rate_mu}).

\item Find the optimal direction (i.e. the maximum eigenvector $\mathbf{\Phi}_\circ$) that maximizes the projection 
\begin{eqnarray}
\mathbf{\Phi}_{o}=\arg\underset{\mathbf{\Phi}}{\max}\quad\mathbf{\Phi}\mathbf{Q}_o\mathbf{\Phi}^H
\end{eqnarray}

\item For all $j$, evaluate $\mathbf{g}_j=\mathbf{h}_j\mathbf{\Phi}_{\circ}$. 
\item Find $j^*=\arg\underset{j}{\max}\quad\|\mathbf{g}_j\|^2$

\item Select the modulation order based on (\ref{modcode}), using $\zeta^{'}_j$

\item  For all possible combinations $\mathcal{G}=\cup_{j}{\mathcal{K}_j}$, evaluate the sum of the users' projection and 
\begin{eqnarray}
\Lambda(\mathcal{K})=\sum_{j\in\mathcal{K}_j\subset \mathcal{G}}\mathbf{g}_j.
\end{eqnarray}

\item Select the subset of users
for all possible of users combinations and find the maximum that has the highest projection
\begin{eqnarray}
\mathcal{K}^*_j=\arg\underset{\mathcal{K}_j}{\max}\quad\|\Lambda (\mathcal{K}_j)\|^2
\end{eqnarray}

\item Evaluate the respective power minimization problem SNR target values.
\begin{eqnarray}
\iota_{j}=\frac{\log_2(\|\mathbf{g}_j\|^2)}{\sum_{i\in\mathcal{K}^*_j}\log_2(\|\mathbf{g}_i\|^2)}
\end{eqnarray}

\begin{eqnarray}
\zeta^{'}_j=\begin{cases}\|\mathbf{g}_j\|^2,\text{if}\quad \|\mathbf{g}_j\|^2\leq
\iota_{j}\|\mathbf{g}^{*}_j\|^2\\
\iota_{j}\|\mathbf{g}^{*}_j\|^2,\text{if}\quad\|\mathbf{g}_j\|^2\geq
\iota_{j}\|\mathbf{g}^{*}_j\|^2 \end{cases}
\end{eqnarray}

\item Solve the related power minimization problem 

\hspace{-0.7cm}\begin{eqnarray}\nonumber
&\underset{\mathbf{q}}{\min}&\mathbf{q}^H\mathbf{q}\\
&s.t.&\begin{cases}\mathcal{C}1:\frac{\mathbf{h}_i\mathbf{q}-(\mathbf{h}_i\mathbf{q})^H}{2i}=\sigma\sqrt{\zeta^{'}_j}\mathcal{I}\{d_i\},
\forall
i\in \mathcal{K}^*_i\\
\mathcal{C}2:\frac{\mathbf{h}_i\mathbf{q}+(\mathbf{h}_i\mathbf{q})^H}{2}=\sigma{\sqrt{\zeta^{'}_i}}\mathcal{R}\{d_i\},\forall
i\in\mathcal{K}^*_i.\end{cases}
\end{eqnarray}

\item Then, scale $\mathbf{q}$ such that $\|\mathbf{q}\|^2=P$. 
\end{enumerate}\\
\hline
\end{tabular}
\end{center}
\normalsize
\section{Algorithms complexity} 
The complexity of the proposed algorithms is an important aspect to assess their feasibility. A discussion about the complexity of each algorithms can be summarized as follows:
\begin{itemize}
\item CIMRT requires an SVD to be employed on the channel $\mathbf{H}$ which has the complexity of $4M^2K+8MK^2+9K^3$ and to solve $\frac{K^2-K}{2}$ times a set of two non-linear equations simultaneously.
\item CIPM requires solving $2K$ linear equations simultaneously, which means that it has less complexity than CIMRT
\item CIMM requires solving $2K$ linear equations simultaneously at each bisection iteration. Moreover, the bisection has a complexity of $\log_2(\max r_k)$.
\item CISR-PA requires solving the convex optimization in (\ref{eigenSR}) without $\mathcal{C}_2$ which calls numerical solvers such as SeDuMi. In order to solve this convex optimization, we need to find the eigenvalue decomposition of $\mathbf{H}\mathbf{H}^H$ which has the complexity of $4M^2K+(8M+1)K^2+9K^3$.  Moreover, this algorithm requires solving $K$ linear equations simultaneously.
\item CISR-G requires solving the convex optimization (\ref{sum_rate_mu})
 and finding the eigenvector associated with maximum eigenvalue. Furthermore, it needs to search all possible combinations $\sum^K_{i}{K\choose i}$ to select the most suitable subset of users to serve in coherence time. Finally, we need to solve $2K$ linear equations simultaneously.
\end{itemize}
\vspace{-0.3cm}
\section{Numerical results}\label{sim}
\vspace{-0.15cm}
In order to assess the performance of the proposed transmissions schemes, Monte-Carlo simulations of the different algorithms have been conducted to
study the performance of the proposed techniques and compare to the state
of the art techniques. The adopted channel model is assumed
to be 
\begin{eqnarray}
\mathbf{h}_k\sim\mathcal{CN}(0,\sigma^2).
\end{eqnarray}
We define the energy efficiency metric as follows 
\begin{eqnarray}
\label{metric}
&\eta&=\sum^K_{j=1}\frac{R_j}{P},\\\nonumber
\end{eqnarray}\vspace{-0.1cm}
where
\begin{eqnarray}\nonumber
&R_j&=\log_2(1+\zeta_j).
\end{eqnarray}
The motivation of using this metric is the fact that CRZF and CIMRT are achievable constructive techniques and cannot be designed based on optimization problems. For the
sake of fairness in comparison, we use the metric in (\ref{metric}).
\begin{center}
\begin{table}
\hspace{-0.02cm}\begin{tabular}{|p{1.3cm}|p{5.2cm}|p{1.0cm}|}
\hline
Acronym&Technique&equation\\
\hline
CIZF&Constructive Interference Zero Forcing&\ref{CIZF}\\
\hline
CIMRT&Constructive Interference Maximum Ratio Transmissions&\ref{CIMRT},
\textbf{A1}\\
\hline
CIPM& Constructive Interference- Power Minimization&\ref{CIPM}\\
\hline
CIMM& Constructive Interference-Maximization the minimum SINR&\textbf{A2}\\
\hline
CISR-G& Constructive Interference-Sum Rate maximization with Greedy approach&\textbf{A3}\\
\hline
CISR-PA& Constructive Interference-Sum Rate maximization with phase alignment&\textbf{A4}\\
\hline
GE& Genie aided upperbound&\ref{pr}\\ 
\hline
Multicast&Optimal Multicast &\ref{powm1}\\
\hline
\end{tabular}
\vspace{0.2cm}
\caption{Summary of the proposed algorithms, their related acronyms, and
their related equations and algorithms}
\end{table}
\end{center}
\vspace{-0.5cm}
\begin{figure}[h]
\vspace{-0.3cm}
\begin{center}
\vspace{-0.15cm}
\hspace{-0.5cm}\includegraphics[scale=0.56]{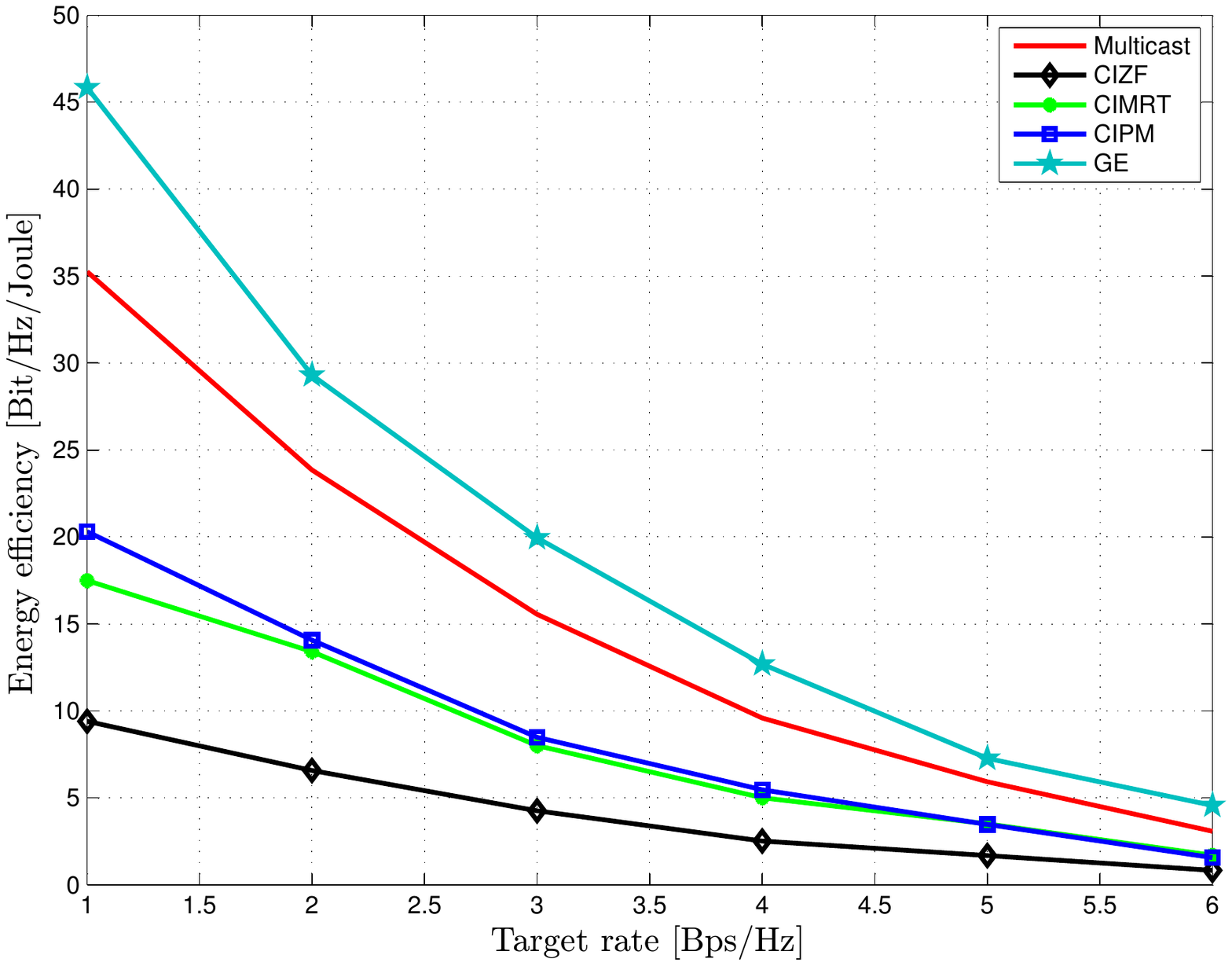}
\vspace{-0.3cm}
\caption{\label{engeff}Energy efficiency vs target rate. }
\end{center}
\end{figure}
In Fig. (\ref{engeff}), we depicted the performance of the proposed techniques from energy efficiency
perspective with respect to the target SNR. We assume the scenario of
$M=3$, 
$K=2$. Since there is no phase alignment constraints
in the unconstrained multicast, it is anticipated that optimal multicast
achieves the highest energy efficiency in the system and this is confirmed
by simulation. Moreover, the gap between the genie aided theoretical bound
and multicast reduces with increasing the target rate (i.e. modulation order
for the genie aided). On contrary, CIZF shows inferior performance in comparison with all
depicted techniques. It has already been proven that CIZF 
outperforms the conventional techniques like minimum mean square error (MMSE)
beamforming and zero forcing beamforming (ZFB) \cite{Christos}. In comparison with other depicted techniques, it can be concluded that
the proposed constructive interference with the CIPM has a better energy efficiency in comparison with CIZF.
This can be explained by the channel inversion step in CIZF which wastes energy in
decoupling the effective users' channels and then exploit the
interference among the multiuser streams. Moreover, it can be deduced that
CIMRT has a very close performance to CIPM especially at high targets. CIMRT
outperforms CIZF at expense of complexity.

The comparison among optimal multicast, CIPM and genie-aided bound is illustrated in Fig. (\ref{powconsg}).
 The assumed scenario $K=2$, $M=2$. It can concluded that the power consumption gap between the optimal multicast and CIPM is fixed for all target rates. This relation holds also for the gap between the genie-aided upperbound and CIPM.
 
\vspace{-0.3cm}
\begin{figure}[h]
\vspace{-0.1cm}
\begin{center}

\hspace{-0.4cm}\includegraphics[scale=0.56]{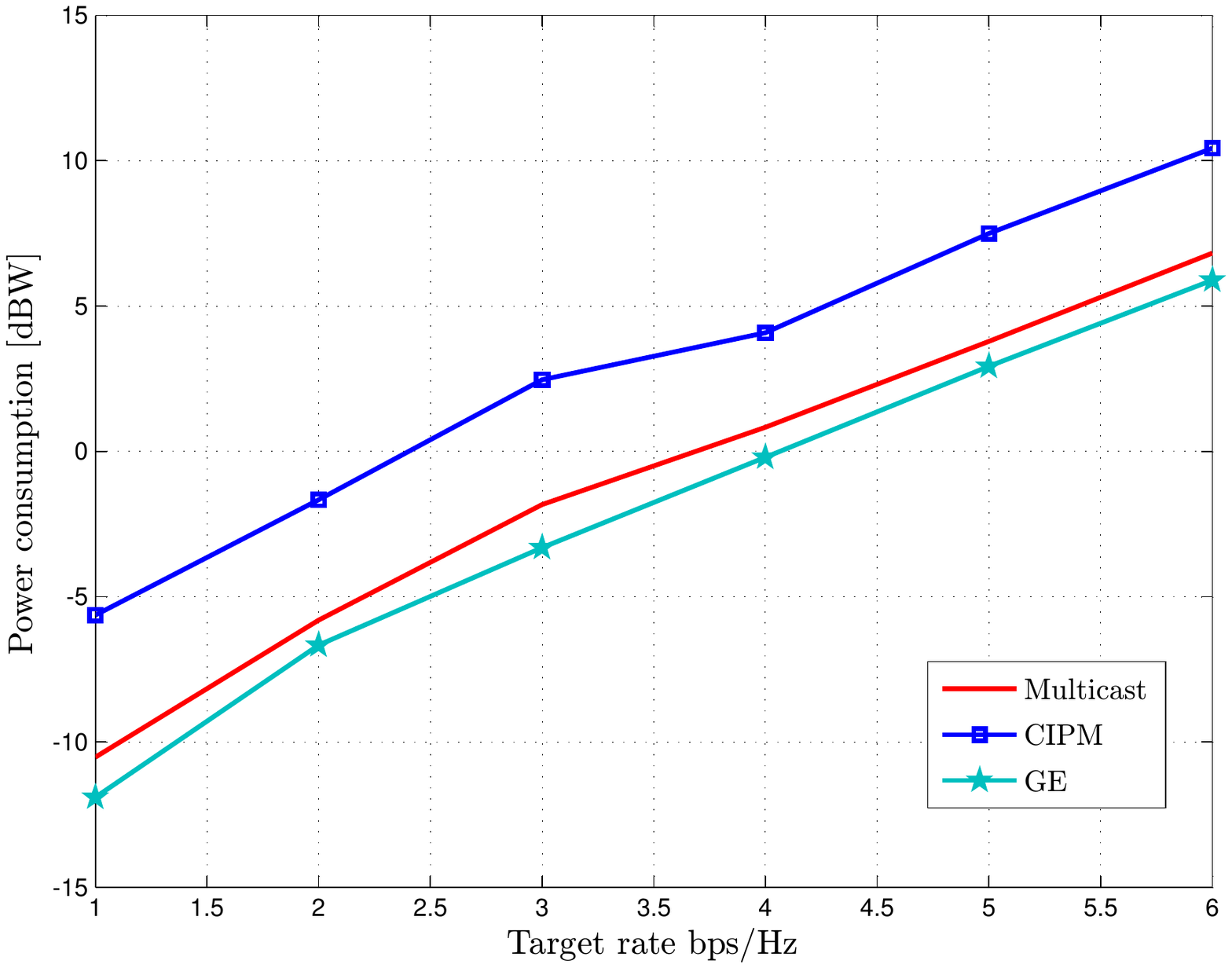}
\vspace{-0.3cm}
\caption{\label{powconsg}\vspace{-0.1cm}Power consumption vs target rate. }
\end{center}
\end{figure}

%
%


\vspace{-0.1cm}
\begin{figure}[h]
\vspace{-0.1cm}
\begin{center}
\vspace{-0.1cm}
\includegraphics[scale=0.56]{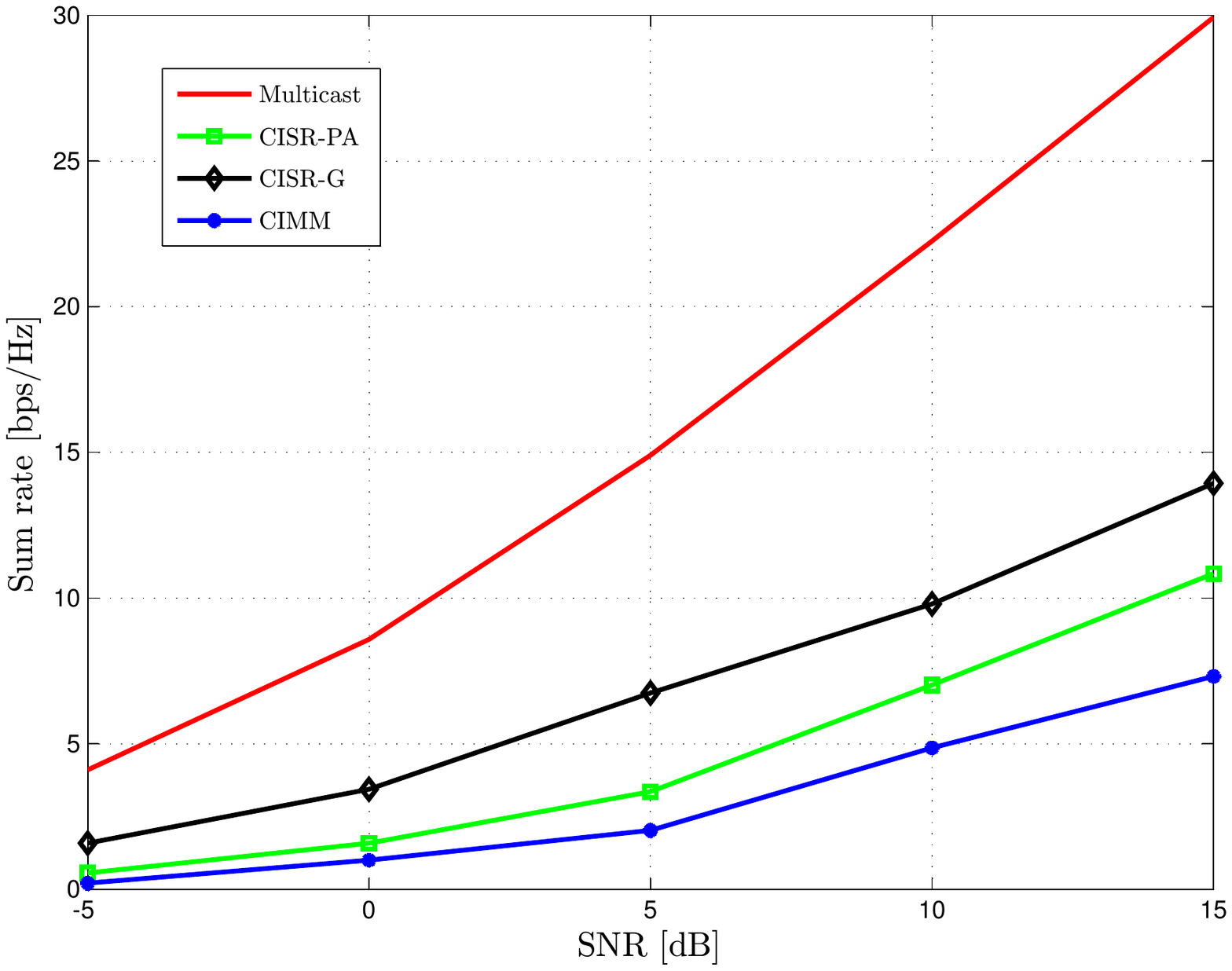}
\vspace{-0.5cm}
\caption{\label{suml} Sum rate vs transmit power. }
\end{center}
\end{figure}
The sum rate performance is illustrated in Fig. (\ref{suml}) in the low-mid
SNR regime. In this scenario,
we consider $K=5$, $M=5$ with equal weights for the sum rate and max-min problems. It can be noted that the
sum rate of algorithm 3, which is implemented in figure as (CISR-G) outperforms
the phase alignment algorithm. It can be concluded that at low SNR, it is
better to preselect the users that have suitable channels to work together. 
In the constructive interference scenario, we tend to select the users whose
channels
are co-linear which opposes the conventional multiuser MISO techniques.
However, for the same scenario, in the high SNR regime which is depicted
in Fig. (\ref{CISR-H}), the phase alignment
algorithm (CISR-PA) shows a better performance than (CISR-G), this means that
it is better not to preselect the users and serve all $K$ users. The performance gap
between the two algorithms increases with SNR. The resulted loss of finding
all the phases that grant the symbol detection by all users has less
effect
on the system performance in comparison to switching off a few users. One should keep in mind that in all scenarios multicast
is just an upper bound and is incapable of delivering different messages to different user. The difference in power consumption is
anticipated since the sum rate problem does not take into the account the
user with the weakest SNR. Finally, it can be concluded that the sum rate for the fairness achieving algorithm (CIMM) is less than (CISR-G) and (CISR-PA) in the low SNR regime. While for high SNR, this fact changes CIMM performs better than (CISR-G) and worse than (CISR-PA).  \\
 
\vspace{-0.3cm}
\begin{figure}[h]
\vspace{-0.4cm}
\begin{center}
\vspace{-0.1cm}
\includegraphics[scale=0.56]{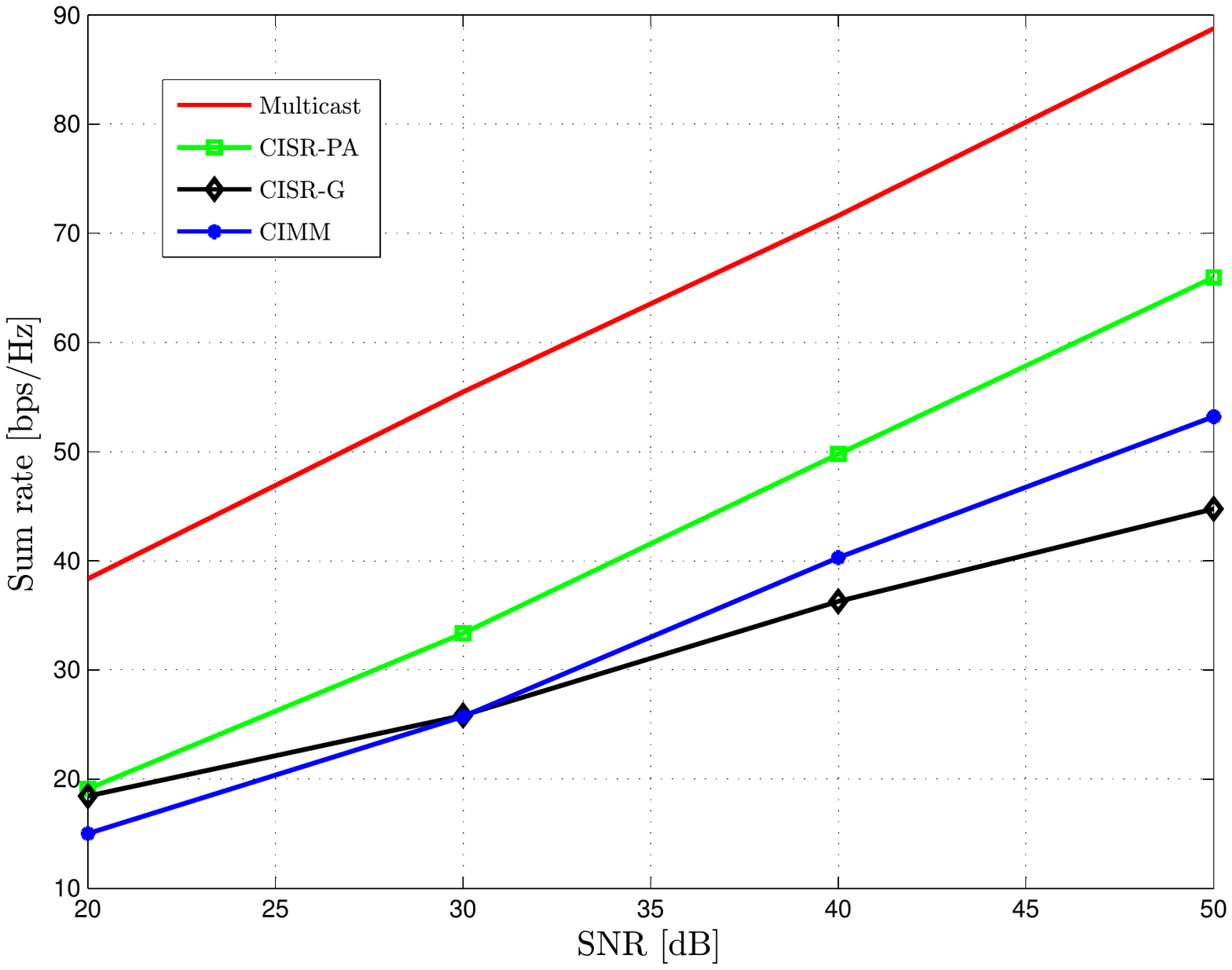}
\vspace{-0.5cm}
\caption{\label{CISR-H}Sumrate vs transmit power. }
\end{center}
\end{figure}

\vspace{-0.3cm} 
\section{conclusions}
In this paper, we exploit the CSI and DI to constructively correlate the transmitted
symbols in symbol-level precoding. This enables interference exploitation among the multiuser transmissions assuming M-PSK modulation. Based on the idea of correlating the transmitted vectors, the connection between the constructive interference precoding and
multicast precoding is characterized. We present several constructive interference designs from different perspective: minimizating the transmitted power while granting certain SNR thresholds for all users, maximizating the fairness among the users, and boosting the sum rate  with fixed transmit power. From the results, it can
be concluded that the max-min SINR problem is related to the power minimzation
problem. Moreover, we tackle the sum rate maximization problem and propose
heuristic solutions to solve the problem. From the simulations, it can be concluded that the CIPM has a fixed transmit power gap with respect to mulicast at different target rates. The sum rate maximization heuristic algorithms vary according to the SNR; CISR-G works very well at low SNR and this changes at high SNR while CISR-PA performs well at high SNR and this pattern changes at low SNR. 

\end{document}